\documentclass[letterpaper, 10 pt, conference]{ieeeconf}  % Comment this line out if you need a4paper

\IEEEoverridecommandlockouts                              % This command is only needed if 
\usepackage{amsfonts}
\usepackage{amstext}
\usepackage{mathtools}
\usepackage{cases}
\usepackage{enumerate}
\usepackage{subfigure}
\usepackage{amssymb}
\usepackage{float}
\usepackage{xcolor}
\usepackage{comment}

\newtheorem{thm}{\textbf{\text{Theorem}}}
\newtheorem{lem}{\textbf{\text{Lemma}}}
\newtheorem{pro}{\textbf{\text{Proposition}}}
\newtheorem{Definition}{\textbf{\text{Definition}}}

\newtheorem{assumption}{\textbf{\text{Assumption}}}
\newtheorem{rmk}{\textbf{\text{Remark}}}
\newtheorem{cnd}{\textbf{\text{Condition}}}

\newcommand{\ie}{\textit{i.e.}}
\newcommand{\eg}{\textit{e.g.}}

\title{\LARGE \bf
Distributed Attitude Estimation for Multi-agent Systems on $SO(3)$}
\author{Mouaad Boughellaba and Abdelhamid Tayebi% <-this % stops a space
\thanks{This work was supported by the National Sciences and Engineering Research Council of Canada (NSERC), under the grants NSERC-DG RGPIN 2020-06270. Preliminary results of this paper have been presented in \cite{Mouaad_ACC23}.} 
\thanks{The authors are with the Department of Electrical Engineering, Lakehead University, Thunder Bay, ON P7B 5E1, Canada \tt\small \{mboughel,atayebi\}@lakeheadu.ca.} 
}

\begin{document}

\maketitle
\thispagestyle{empty}
\pagestyle{empty}

%%%%%%%%%%%%%%%%%%%%%%%%%%%%%%%%%%%%%%%%%%%%%%%%%%%%%%%%%%%%%%%%%%%%%%%%%%%%%%%%
\begin{abstract}
 We consider the problem of distributed attitude estimation of multi-agent systems, evolving on $SO(3)$, relying on individual angular velocity and relative attitude measurements. The interaction graph topology is assumed to be an undirected tree. First, we propose a continuous nonlinear distributed attitude estimation scheme with almost global asymptotic stability guarantees. Thereafter, we proceed with the \textit{hybridization} of the proposed estimation scheme to derive a new hybrid nonlinear distributed attitude estimation scheme enjoying global asymptotic stabilization of the attitude estimation errors to a common constant orientation. In addition, the proposed hybrid attitude estimation scheme is used to solve the problem of pose estimation of $N$-vehicles navigating in a three-dimensional space, with global asymptotic stability guarantees, where the only available measurements are the local relative bearings and  the individual linear velocities. Simulation results are provided to illustrate the effectiveness of the proposed estimation schemes.
\end{abstract}

\section{Introduction}
In the last few decades there has been a growing interest in the development of cooperative estimation and control schemes for multi-agent autonomous systems such as satellites, unmanned aerial vehicles (UAVs) and marine vehicles. %\cite{VIEGAS201572, TRAN2020109125, VIEGAS2012443}. 
The distributed cooperative estimation problem for multi-agent autonomous vehicles, which is the topic of the present paper, consists in estimating some (or all) of the agents' states (position, orientation, velocity) using local information exchange between neighboring agents. A particular problem of great importance in this field, namely the distributed cooperative attitude estimation, consists of estimating the agents' attitudes using some absolute individual measurements and some inter-agent (relative) measurements according to a predefined interaction graph topology between the agents involved in the group. The importance of this problem stems from the fact that the available absolute individual measurements (assumed in this paper) are not enough to allow each agent to estimate its orientation independently from other agents.
%schemes can be categorized by the choice of the attitude parametrization, such as Rodrigues parameters, Euler-angles, and unit quaternion. However, all these parameterizations fail to represent the agent's orientation globally and uniquely, which may result in estimation schemes that suffer from the singularity issue or the so-called unwinding phenomenon. Therefore, a number of research works proposed distributed cooperative attitude estimation schemes designed directly on the Special Orthogonal group $SO(3)$, since the orientations of agents are uniquely and globally represented by a rotation matrix in $SO(3)$. 
It is well known that the only representation that represents the attitude of a rigid body uniquely and globally is the rotation matrix belonging to the special orthogonal group $SO(3)$ which is a smooth manifold with group properties (\textit{i.e.,} matrix Lie group).
%is group structure enjoying the matrix Lie group structure which allows the design and analysis of attitude estimation schemes using the well established tools from differential geometry.
%the well established framework \cite{bullo2019geometric}. 
Since $SO(3)$ is boundaryless odd-dimentioned compact manifold, it is non-diffeomorphic to any Euclidean space, and as such, it is not possible to achieve global stability results with time-invariant continuous vector fields on $SO(3)$ \cite{Koditschek,Bhat_SCL2000}, \textit{i.e.,} on top of the desired attractive equilibrium there are other undesired equilibria. Since the use of classical cooperative schemes (on Euclidean spaces) for systems evolving on smooth manifolds is not trivial, appropriate cooperative control and estimation techniques needed to be developed. As such, some consensus-based attitude synchronization schemes on $SO(3)$ have been proposed in the literature (see, for instance, \cite{Tron_CDC2012, Tron_TAC2012, Markdahl_TAC2020, SARLETTE2009572, SARLETTE20072232, Alain_SIAM}). 
%These solutions can be categorized into two main categories: the intrinsic approaches and the extrinsic approaches. In the intrinsic approach, the distributed attitude control laws are designed based on the intrinsic geometric properties of $SO(3)$. On the other hand, the extrinsic approaches rely embedding the rotation manifold $SO(3)$ in the Euclidean space $\mathbb{R}^9$, so that classical consensus algorithms can be employed, and then projecting the states, via an appropriate projection operator, on the rotation manifold $SO(3)$.\\ 
Motivated by the attitude synchronization schemes mentioned in the above references, some distributed cooperative attitude estimation schemes, designed directly on $SO(3)$, have been proposed in the literature. 
For instance, the authors in \cite{Tron_TAC2014} proposed a distributed attitude localization scheme for camera sensor networks by reshaping the cost function, given in \cite{Tron_CDC2011}, such that the only stable equilibria of the proposed scheme are global minimizers. However, this estimation scheme guarantees only almost global asymptotic stability. Also, in a series of papers \cite{lee_2016, lee_auto2016, Tran_CDC2018, lee_2019}, the authors proposed some attitude estimation schemes relying on classical (Euclidean) consensus algorithms along with the Gram-Schmidt orthonormalization procedure. Recently, these attitude estimation algorithms have been extended to deal with time-varying orientations in n-dimensional Euclidean spaces \cite{Tran_TCNS2019, Tran_TCNS2020}. Note that the attitude estimation schemes based on the Gram-Schmidt orthonormalization procedure may lead to problems when the estimated matrix (that does not belong to $SO(3)$) is singular. A more recent work \cite{li_2020} suggests a similar algorithm as the one in  \cite{lee_2016, Tran_CDC2018}, but without Gram-Schmidt orthonormalization. As a result, the algorithm provides estimates of the agents' orientations only when the time tends to infinity and not for all times, which makes the algorithm inappropriate for control applications requiring instantaneous orientations for feedback.\\
In this paper, we propose two distributed attitude estimation schemes on $SO(3)$ for a group of rigid body agents under an undirected, connected and acyclic graph topology. Each agent measures its own angular velocity in the respective body-frame, measures the relative orientation with respect to its neighbours, and receives information from its neighbors. Note that both estimation schemes lead to attitude estimates up to a constant orientation which can be determined in the presence of a leader in the group (knowing its absolute orientation). Furthermore, as an application, we design a hybrid distributed position estimation law that uses the estimated attitudes, provided by the hybrid distributed attitude observer, the local relative (time-varying) bearing information, and the individual linear velocities. The contribution of this work can be summarized as follows:
\begin{enumerate}
    \item Inspired by the consensus optimization framework on manifolds introduced in \cite{Sarlette2009}, we propose a continuous distributed attitude estimation scheme on $SO(3)$. Moreover, we provide a rigorous stability analysis that shows that the proposed continuous attitude observer enjoys almost global asymptotic stability. Compared to the existing works, such as \cite{lee_2016,lee_2019,Tran_TCNS2019,Tran_TCNS2020}, the proposed continuous distributed attitude observer, on top of being designed directly on $SO(3)$ and endowed with stability (not only convergence) results, is much simpler and does not require any auxiliary matrices and orthonormalization procedures (\eg, Gram-Schmidt orthonormalization), which may complicate the implementation of the observer and add extra computational overhead.
    %\item We design a generic multi-agent switching mechanism 
    \item To overcome the topological obstruction that precludes global asymptotic stability on $SO(3)$ with smooth vector fields, we propose a new hybrid distributed attitude observer enjoying global asymptotic stability. In contrast to the conference version of this work \cite{Mouaad_ACC23}, we have developed a systematic and generic approach for the design of this distributed hybrid attitude estimation scheme. This approach has also facilitated the introduction of a new generic multi-agent switching mechanism that can potentially be applied to address other challenging problems, such as global synchronization on $SO(3)$.   
    \item Finally, we propose a bearing-based hybrid distributed pose estimation scheme guaranteeing global pose estimation of the $N$-agent system up to a constant translation and orientation. Unlike our previous work \cite{Mouaad_LCSS23}, which considers static positions and time-varying orientations, and guarantees almost global asymptotic stability, the present work deals with time-varying positions and orientations with global asymptotic stability guarantees. To the best of the authors' knowledge, there is no work in the literature that addresses the problem of distributed pose estimation, with global asymptotic stability guarantees, considering time-varying positions and orientations with local time-varying bearing measurements.
\end{enumerate}
A preliminary version of this work was presented in \cite{Mouaad_ACC23}, excluding the continuous distributed attitude estimation scheme and the bearing-based distributed hybrid pose estimation scheme discussed in the present paper.

The remainder of this paper is organized as follows: Section \ref{s2} provides some preliminaries necessary for this paper. In Section \ref{s3}, we formulate the distributed attitude estimation problem for multi-agent systems. Section \ref{s4} presents our proposed continuous and hybrid distributed attitude estimation schemes with a rigorous stability analysis. In section \ref{s5}, we address the problem of distributed pose estimation of the $N$-agent system formation, evolving in a three-dimensional space, using our proposed hybrid distributed attitude estimation scheme. Simulation results and concluding remarks are presented in Section \ref{s6} and Section \ref{s7}, respectively.

%First, we propose a continuous distributed attitude estimation scheme on $SO(3)$, inspired from \cite{Sarlette2009}, endowed with almost global asymptotic stability. Thereafter, aiming at strengthening the stability results our first estimation scheme, we propose a new hybrid distributed attitude estimation algorithm enjoying global asymptotic stability. Note that both estimation schemes lead to attitude estimates up to a constant orientation which can obtained in the presence of a leader in the group (knowing its absolute orientation). Furthermore, as an application, we design a distributed position estimation law that uses the estimated attitudes, provided by the hybrid distributed attitude observer, the local relative (time-varying) bearing information, and the individual linear velocities. This design guarantees global estimation of the $N$-agent system formation moving in a three-dimensional space.

\section{Preliminaries}\label{s2}

\subsection{Notations}
\noindent The sets of real numbers and the $n$-dimensional Euclidean space  are denoted by $\mathbb{R}$ and $\mathbb{R}^n$, respectively. The set of unit vectors in $\mathbb{R}^n$ is defined as $\mathbb{S}^{n-1}:=\{x\in \mathbb{R}^n~|~x^T x =1\}$. Given two matrices $A$,$B$ $\in \mathbb{R}^{m\times n}$, their Euclidean inner product is defined as $\langle \langle A,B \rangle \rangle=\text{tr}(A^T B)$. The Euclidean norm of a vector $x \in \mathbb{R}^n$ is defined as $||x||=\sqrt{x^T x}$, and
the Frobenius norm of a matrix $A \in \mathbb{R}^{n\times n}$ is given by $||A||_F=\sqrt{\text{tr}(A^T A)}$. The matrix $I_n \in \mathbb{R}^{n \times n}$ denotes the identity matrix, and $\textbf{1}_n=[1\hdots1]^T \in \mathbb{R}^n$. Consider a smooth manifold $\mathcal{Q}$ with $\mathcal{T}_x \mathcal{Q}$ being its tangent space at point $x \in \mathcal{O}$. Let $f: \mathcal{Q} \rightarrow \mathbb{R}_{\geq 0}$ be a continuously differentiable real-valued function. The function $f$ is a potential function on $\mathcal{Q}$ with respect to set $\mathcal{B}\in \mathcal{Q}$ if $f(x)=0$, $\forall x \in \mathcal{B}$, and $f(x) > 0$, $\forall x\notin \mathcal{B}$. The gradient of $f$ at $x \in \mathcal{Q}$, denoted by $\nabla_x f(x)$, is defined as the unique element of $\mathcal{T}_x \mathcal{Q}$ such that \cite{Mahony_book_OAMM}:
\begin{equation}
    \dot f(x)=\langle \nabla_x f(x), \eta\rangle_x, ~~~~~ \forall \eta \in \mathcal{T}_x \mathcal{Q},\nonumber
\end{equation}
where $\langle~,~\rangle_x:\mathcal{T}_x \mathcal{Q} \times \mathcal{T}_x \mathcal{Q} \rightarrow \mathbb{R}$ is Riemannian metric on $\mathcal{Q}$. The point $x \in \mathcal{Q}$ is said to be a critical point of $f$ if $\nabla_x f(x)=0$.\\
The attitude of a rigid body is represented by a rotation matrix $R$ which belongs to the special orthogonal group $SO(3):= \{ R\in \mathbb{R}^{3\times 3} | \hspace{0.1cm}\text{det}(R)=1, R^TR=I_3\}$. The $SO(3)$ group has a compact manifold structure and its tangent space is given by $\mathcal{T}_RSO(3):=\{R \hspace{0.1cm}\Omega \hspace{0.2cm} | \hspace{0.2cm} \Omega \in \mathfrak{so}(3)\}$ where $\mathfrak{so}(3):=\{ \Omega \in \mathbb{R}^{3\times 3} | \Omega^T=-\Omega\}$ is the Lie algebra of the matrix Lie group $SO(3)$. The map $[.]^{\times}: \mathbb{R}^3 \rightarrow \mathfrak{so}(3)$ is defined such that $[x]^\times y=x \times y$, for any $x,y \in \mathbb{R}^3$, where $\times$ denotes the vector cross product on $\mathbb{R}^3$. The inverse map of $[.]^{\times}$ is $\text{vex}: \mathfrak{so}(3) \rightarrow \mathbb{R}^3$ such that $\text{vex}([\omega]^\times)=\omega$, and $[\text{vex}(\Omega)]^\times=\Omega$ for all $\omega \in \mathbb{R}^3$ and $\Omega \in \mathfrak{so}(3)$. Also, let $\mathbb{P}_a : \mathbb{R}^{3\times 3} \rightarrow \mathfrak{so}(3)$ be the projection map on the Lie algebra $\mathfrak{so}(3)$ such that $\mathbb{P}_a(A):=(A-A^T)/2$. Given a 3-by-3 matrix $C:=[c_{ij}]_{i,j=1,2,3}$, one has  $\psi(C) := \text{vex} \circ \mathbb{P}_a (C)=\text{vex}(\mathbb{P}_a(C))=\frac{1}{2}[c_{32}-c_{23},c_{13}-c_{31},c_{21}-c_{12}]^T$. For any $R\in SO(3)$, the normalized Euclidean distance on $SO(3)$, with respect to the identity $I_3$, is defined as $|R|_I^2:=\frac{1}{4}\text{tr}(I_3-R)$ $\in[0,1]$. The angle-axis parameterization of $SO(3)$, is given by $\mathcal{R}_\alpha(\theta, v):=I_3+sin\hspace{0.05cm}\theta \hspace{0.2cm}[v]^\times + (1-cos\hspace{0.05cm}\theta)([v]^\times)^2$, where $v\in \mathbb{S}^2$ and  $\theta \in \mathbb{R}$ are the rotational axis and angle, respectively. The orthogonal projection map $P:\mathbb{R}^3\rightarrow\mathbb{R}^{3\times3}$ is defined as $P_x:=I_3-\frac{x x^T}{||x||^2}$, $\forall x\in\mathbb{R}^3\setminus\{0_3\}$. One can verify that $P_x^T=P_x$, $P_x^2=P_x$, and $P_x$ is positive semi-definite. The Kronecker product of two matrices A and B is denoted by $A \otimes B$. The n-dimensional orthogonal group, denoted by $O(n)$, is defined as $O(n):=\{A \in \mathbb{R}^{n\times n}:~A A^T=A^T A= I_n\}$. 
\subsection{Graph Theory}
Consider a network of $N$ agents. The interaction topology between the agents is described by an undirected graph $\mathcal G = (\mathcal V,\mathcal E)$ where $\mathcal V=\{1,...,N\}$ and $\mathcal E \subseteq \mathcal V \times \mathcal V $ represent the vertex (or agent) set and the edge set of graph $\mathcal{G}$, respectively. In undirected graphs, the edge $(i,j) \in \mathcal E$ indicates that agents $i$ and $j$ interact with each other without any restriction on the direction, which means that agent $i$ can obtain information (via communication, measurements, or both) from agent $j$ and vice versa. The set of neighbors of agent $i$ is defined as $\mathcal N_i = \{j \in \mathcal V : (i,j) \in \mathcal E \}$. The undirected path is a sequence of edges in an undirected graph. An undirected graph is called connected if there is an undirected path between every pair of distinct agents of the graph. An undirected graph has a cycle if there exists an undirected path that starts and ends at the same agent \cite{Ren_book}. An acyclic undirected graph is an undirected graph without a cycle. An undirected tree is an undirected graph in which any two agents are connected by exactly one path (\ie, an undirected tree is an undirected, connected, and acyclic graph). Consider an undirected tree with an arbitrarily assigned orientation to each edge. Let $M=|\mathcal{E}|$ and $\mathcal{M}=\{1,\hdots, M\}$ be the total number of edges and the set of edge indices, respectively. The incidence matrix, denoted by $H\in\mathbb{R}^{N\times M}$, is defined as follows \cite{BAI20083170}:
\begin{equation}\label{h_matrix}
   H:=[h_{ik}]_{N\times M} \hspace{0.4cm} \text{with} \hspace{0.2cm} h_{ik}=\begin{cases}
      +1 & k\in\mathcal{M}_i^+\\
      -1 & k\in\mathcal{M}_i^-\\
      0 & \text{otherwise}
    \end{cases} \nonumber,    
\end{equation}   
where $\mathcal{M}_i^+ \subset \mathcal{M}$ denotes the subset of edge indices in which agent $i$ is the head of the oriented edges and $\mathcal{M}_i^- \subset \mathcal{M}$ denotes the subset of edge indices in which agent $i$ is the tail of the oriented edges. For a connected graph, one verifies that $H^T\textbf{1}_N=0$ and $\mbox{rank}(H)=N-1$. Moreover, the columns of $H$ are linearly independent if the graph is an undirected tree.

\subsection{Hybrid Systems Framework}
A hybrid system consists of continuous dynamics called flows and discrete dynamics called jumps. Given a manifold $\mathcal{Y}$ embedded in $\mathbb{R}^n$, according to \cite{Goebel_automatica_2006, Goebel_ieee_magazine,goebel2012hybrid}, the hybrid system dynamics, denoted by $\mathcal{H}$, are represented by the following compact form, for every $y\in \mathcal{Y}$:
\begin{align}\label{H_sys}
    \mathcal{H}:
    \begin{cases}
       \dot{y} \in F(y) &\quad y\in \mathcal{F}\\
       y^+ \in G(y) & \quad y\in \mathcal{J}
    \end{cases}
\end{align}
where $F$ and $G$ represent the flow and jump maps, respectively, which govern the dynamics of the state $y$ through a continuous flow (if $y$  belong to the flow set $\mathcal{F}$) and a discrete jump (if  $y$ belong to the jump set $\mathcal{J}$). According to the nature of the hybrid system dynamics, which allows continuous flows and discrete jumps, the solutions of the hybrid system are parameterized by $t \in \mathbb{R}_{\geq 0}$ to indicate the amount of time spent in the flow set and $j \in \mathbb{N}$ to track the number of jumps that occur. The structure that represents this parameterization, which is known as a \textit{hybrid time domain}, is a subset of $\mathbb{R}_{\geq 0} \times \mathbb{N}$ and is denoted as \textup{dom}~$y$. If the solution of the hybrid system $\mathcal{H}$ cannot be extended by either flowing or jumping, it is called \textit{maximal}, and \textit{complete} if its domain dom x is unbounded.

\section{Problem Statement}\label{s3}
\noindent Consider $N$-agent system governed by the following rotational kinematic equation:
\begin{equation}\label{R_dynamics}
    \dot{R}_i = R_i[\omega_i]^{\times},
\end{equation}
where $R_i \in SO(3)$ represents the orientation of the body-attached frame of agent $i$ with respect to the inertial frame, and $\omega_i\in \mathbb{R}^3$ is the angular velocity of agent $i$ measured in the body-attached frame of the same agent. The measurement of the relative orientation between agent $i$ and agent $j$ is given by
\begin{equation}\label{measurement_model_R}
    R_{ij} := R_i^T R_j,
\end{equation}
where $(i,j)\in\mathcal E$. Note that according to the kinematic equation \eqref{R_dynamics}, the agents orientations vary with time. However, for the sake of simplicity, the time argument $t$ has been omitted from the above expression. Let the graph $\mathcal G$ describe the interaction between agents (the relative measurements and communication). Now, consider the following assumptions:

\begin{assumption}\label{measurement_ass}
  Each agent $i \in \mathcal V$ measures the relative orientations $R_{ij}$ with respect to its neighboring agents $j\in \mathcal{N}_i$. In addition, each agent can also share information, through communication, with its neighbors. 
\end{assumption}

\begin{assumption}\label{graph_ass}
  The interaction graph $\mathcal G$ is assumed to be an undirected tree. 
\end{assumption}

\begin{assumption}\label{measurement_avialble_w}
  The body-frame angular velocity of each agent is bounded and available.
\end{assumption}
The problem under consideration can be formulated as follows:\\
Consider a network of $N$ agents rotating according to the kinematics (\ref{R_dynamics}) where the relative attitude measurements \eqref{measurement_model_R} are available according to the graph topology $\mathcal{G}$. Design a distributed attitude estimation scheme endowed with global asymptotic stability guarantees, under Assumptions \ref{measurement_ass}-\ref{measurement_avialble_w}.
\begin{rmk}
    Since only relative attitude measurements are available, it should be understood from the above problem statement that the goal is to globally estimate the orientation of each agent up to a common constant orientation. This common constant orientation can be determined if at least one agent has access to its own orientation.
\end{rmk}

\section{Distributed Attitude Estimation}\label{s4}
For $i\in\mathcal{V}$, we propose the following attitude observer on $SO(3)$:
\begin{equation}\label{observer_dynamics}
    \dot{\hat{R}}_i=\hat{R}_i[\omega_i-k_R\hat{R}_i^T \sigma_i]^\times,
\end{equation}
where $k_R>0$, $\hat{R}_i \in SO(3)$ is the estimate of $R_i$, and $\sigma_i \in \mathbb{R}^3$ is the correcting term that will be designated later. Let $\tilde R_i:=R_i \hat{R}_i^T$ denote the absolute attitude error of agent $i$. In view of (\ref{R_dynamics}) and (\ref{observer_dynamics}), one has 
\begin{equation}\label{absolute_attitude_error}
    \dot{\tilde R}_i = k_R\tilde R_i [ \sigma_i]^\times.
\end{equation}
Consider an arbitrary orientation of the interaction graph $\mathcal{G}$. For every $(i, j) \in \mathcal{E}$, suppose that agent $i$ and agent $j$ are the head and tail, respectively, of the oriented edge connecting them, indexed by $k$. One can define the relative attitude error between them as $\bar R_k := \tilde R_j^T \tilde R_i$, where $k = \mathcal{M}_i^+ \cap \mathcal{M}_j^- \in \mathcal{M}$. It follows from (\ref{absolute_attitude_error}) that
\begin{align}\label{exp_R_bar}
    \dot{\bar R}_k &= -k_R [\sigma_j]^\times \tilde R_j^T \tilde R_i + k_R \tilde R_j^T \tilde R_i [\sigma_i]^\times \nonumber\\
    &= -k_R [\sigma_j]^\times \bar R_k + k_R \bar R_k [\sigma_i]^\times.
\end{align}
Using the fact that $[R x]^\times=R x^\times R^T$ for every $R \in SO(3)$ and $x \in \mathbb{R}^3$, it follows from \eqref{exp_R_bar} that 
\begin{align}
    \dot{\bar R}_k &= -k_R \bar R_k [\bar R_k^T\sigma_j]^\times  + k_R \bar R_k [\sigma_i]^\times\nonumber\\
    &= k_R \bar R_k \left(-[\bar R_k^T\sigma_j]^\times  +[\sigma_i]^\times \right).
\end{align}
Moreover, since $[x+y]^\times=[x]^\times+[y]^\times$ for every $x, y \in \mathbb{R}^3$, one has
\begin{equation}\label{relative_att_dynamics}
    \dot{\bar R}_k = k_R\bar R_k [\bar \sigma_k]^\times,
\end{equation}
where $\bar \sigma_k :=  \sigma_i - \bar R_k^T \sigma_j$. Note that, for every $i \in \mathcal{V}$ and $ j \in \mathcal{N}_i$, the intersection between the sets $\mathcal{M}_i^+$ and $\mathcal{M}_j^-$ is either a single element (if agent $i$ and agent $j$ are the head and tail, respectively, of the oriented edge connecting them) or an empty set otherwise. Let $\bar \sigma :=[\bar \sigma_1^T, \bar \sigma_2^T, \hdots, \bar \sigma_M^T]^T \in \mathbb{R}^{3M}$ and $\sigma :=[\sigma_1^T, \sigma_2^T, \hdots, \sigma_N^T]^T \in \mathbb{R}^{3N}$. One can verify that \cite{BAI20083170}
\begin{equation}\label{s_bar}
    \bar \sigma = \textit{\textbf{H}}(t)^T \sigma,
\end{equation}
where 
\begin{equation}\label{H_bar}
   \textit{\textbf{H}}(t):=[H_{ik}]_{3N\times 3M} \hspace{0.3cm} \text{with} \hspace{0.3cm} H_{ik}=\begin{cases}
      I_3 & k\in\mathcal{M}_i^+\\
      -\bar R_k & k\in\mathcal{M}_i^-\\
      0 & \text{otherwise}
    \end{cases}.   
\end{equation} 
The matrix $\textit{\textbf{H}}$ inherits some properties from the incidence matrix $H$, such as the adjacency relationships in the graph and the orientation that the graph enjoys. Moreover, in the following lemma, we will introduce an important property that $\textit{\textbf{H}}$ enjoys when Assumption \ref{graph_ass} holds.
\begin{lem}\label{lem_H}
    Consider the matrix $\textit{\textbf{H}}(t)$ obtained from the graph $\mathcal{G}$ with an arbitrary orientation of the edges, satisfying Assumption \ref{graph_ass}. Then, $\forall t\geq 0$, $\textit{\textbf{H}}(t)x=0$  implies $x=0$.
\end{lem}
\begin{proof}
    The proof can be found in \cite{Mouaad_ACC23}.
\end{proof}
\begin{rmk}
    It is important to note that the arbitrary orientation assigned to the graph $\mathcal{G}$ is only a dummy orientation introduced to simplify the design process and analysis of our proposed schemes, and does not change the nature of the interaction graph $\mathcal{G}$ from being an undirected graph. 
\end{rmk}
In the sequel, we propose two attitude estimation schemes through an appropriate design of the correcting term $\sigma_i$. We will start with the continuous version of the observer in the next subsection.

\subsection{Continuous Distributed Attitude Estimation}
Consider the observer given in \eqref{observer_dynamics}, with the  following correcting term:
\begin{align}\label{C_correcting_term}
    \sigma_i&=-\sum_{j\in \mathcal{N}_i} \psi(A\hat{R}_jR_{ij}^T\hat{R}_i^T),
\end{align}
where $i\in\mathcal{V}$ and $A\in \mathbb{R}^{3\times3}$. Let $x:= (\bar R_1, \bar{R}_2, \hdots, \bar R_M)\in \mathcal{S}$, with $\mathcal{S}:=\left(SO(3)\right)^M$. From \eqref{relative_att_dynamics}-\eqref{s_bar} and \eqref{C_correcting_term}, one can derive the following multi-agent closed-loop dynamics:
\begin{equation}\label{continuous_sys}
\dot x = f(x) \hspace{1.5cm}x \in \mathcal{S}, 
\end{equation}
where
\begin{equation}
    f(x) = \left[ {\begin{array}{c} {{f_1(x)}}\\ \vdots\\ {f_M(x)} \end{array}}\right] \hspace{0.5cm} \text{with} \hspace{0.5cm} f_k(x) = k_R\bar{R}_k[\bar \sigma_k]^\times.\nonumber
\end{equation}
Before presenting the main result of this section, we first introduce the following instrumental assumption:
\begin{assumption}\label{A_ass}
    $A \in \mathbb{R}^{3 \times 3}$ is a symmetric and positive definite matrix with three distinct eigenvalues.
\end{assumption}
\noindent The following theorem provides the stability properties of the equilibrium points of \eqref{continuous_sys}.
\begin{thm}\label{theorem1}
     Consider the attitude kinematics (\ref{R_dynamics}) with measurements (\ref{measurement_model_R}) and observer (\ref{observer_dynamics}) together with the correcting term (\ref{C_correcting_term}), where Assumptions \ref{measurement_ass}-\ref{A_ass} are satisfied. Then, the following statements hold:
   \begin{enumerate}[i)]
       \item All solutions of (\ref{continuous_sys}) converge to the following set of equilibria: $\Upsilon : =\{x \in \mathcal{S}: \forall k\in \mathcal{M}, ~\psi(A\bar{R}_k)=0\}$. \label{set_of_equilibrium} 
       %$\mathcal{M}=\mathcal{M}^I \cup \mathcal{M}^\pi$ with $|\mathcal{M}^\pi|>0$.} \label{set_of_equilibrium}
    %\label{set_of_equilibrium}
     %  \begin{align*}& {\Upsilon_2 : = \left({{I_3},0}\right) \cup \left\{ {(\tilde R_i,\tilde \Omega_i) \in SO(3) \times {{\mathbb{R}}^3}:} \right.} \\ & {\left. {\tilde R_i = {{\mathcal{R}}_\alpha }(\pi ,v_i), v_i \in {\mathcal{E}}(K_i),\tilde \Omega_i = 0} \right\}.} \end{align*}
     \item \label{stable_equilibrium} The desired equilibrium set $\mathcal{A}:=\{x \in \mathcal{S}: \forall k\in \mathcal{M}, ~\bar{R}_k=I_3\} \in \Upsilon$, for the closed-loop system (\ref{continuous_sys}), is locally asymptotically stable.\label{A_local_stability}
       \item  The set of all undesired equilibrium points $\Upsilon \setminus \mathcal{A}$ is unstable and the desired equilibrium set $\mathcal{A}$ is \textit{almost globally asymptotically stable}\footnote{ The set $\mathcal{A}$ is said to be almost globally asymptotically stable if it is asymptotically stable, and attaractive from all initial conditions except a set of zero Lebesgue measure.} for the closed-loop system (\ref{continuous_sys}). \label{stability_of_equilibrium}
   \end{enumerate}
\end{thm}
\begin{proof}
It follows from (\ref{C_correcting_term}) that
\begin{equation}\label{correcting_term_tilde}
     \sigma_i = -\sum_{j\in \mathcal{N}_i}\psi(A\tilde R_j^T \tilde R_i).
\end{equation}
Since the graph $\mathcal{G}$ is an undirected graph with an orientation, one can verify that $\mathcal{N}_i=\mathcal{I}_i\cup\mathcal{O}_i$, with $\mathcal{I}_i:=\{j\in \mathcal{N}_i: j \hspace{0.1cm} \text{is the tail of the oriented edge} \hspace{0.1cm} (i,j) \in \mathcal{E}\}$ and $\mathcal{O}_i:=\{j\in \mathcal{N}_i: j \hspace{0.1cm} \text{is the head of the oriented edge} \hspace{0.1cm} (i,j) \in \mathcal{E}\}$. Therefore, it follows from \eqref{correcting_term_tilde} that
\begin{align}
    \sigma_i &= -\left( \sum_{j\in \mathcal{I}_i}\psi(A\tilde R_j^T \tilde R_i)+\sum_{j\in \mathcal{O}_i}\psi(A\tilde R_j^T \tilde R_i)\right)\nonumber\\
    &= -\left( \sum_{j\in \mathcal{I}_i}\psi(A\tilde R_j^T \tilde R_i)-\sum_{j\in \mathcal{O}_i}\psi(\tilde R_i^T \tilde R_j A)\right)\label{eq_1}\\
    &= -\left( \sum_{j\in \mathcal{I}_i}\psi(A\tilde R_j^T \tilde R_i)-\sum_{j\in \mathcal{O}_i}\tilde R_i^T \tilde R_j\psi(A\tilde R_i^T \tilde R_j )\right)\label{eq_2}\\
    &= -\left( \sum_{n\in \mathcal{M}_i^+}\psi(A\bar{R}_n)-\sum_{l\in \mathcal{M}_i^-}\bar{R}_l\psi(A\bar{R}_l)\right)\nonumber\\
    &= -\sum_{k=1}^M H_{ik}\psi(A\bar{R}_k),
\end{align}
where $H_{ik}$ is given in (\ref{H_bar}). Equations (\ref{eq_1}) and (\ref{eq_2}) are obtained using the facts that $\psi(BR)=-\psi(R^TB)$ and $\psi(GR)=R^T\psi(RG)$, $\forall G, B=B^T\in \mathbb{R}^{3\times3} $ and $R\in SO(3)$. Moreover, one can verify that
 \begin{equation}\label{sigma_c}
     \sigma = -\textit{\textbf{H}} \Psi,
 \end{equation}
 where $\Psi:=\left[\psi(A\bar R_1)^T, \psi(A\bar R_2)^T, \hdots, \psi(A\bar R_M)^T\right]^T \in \mathbb{R}^{3M}$. For the sake of simplicity, we write the block matrix $\textit{\textbf{H}}$ without the time argument. On the other hand, consider the following Lyapunov function candidate:
\begin{equation}
    V(x) = \sum_{k=1}^{M} \text{tr}\left(A(I_3-\bar{R}_k)\right),
\end{equation}
which is positive definite on $\mathcal{S}$ with respect to $\mathcal{A}$. Note that {\small $\Psi=\Bigg[\psi\left(\bar R^T_1\nabla_{\bar R_1} V\right)^T,\psi\left(\bar R^T_2\nabla_{\bar R_2} V\right)^T, \hdots,$ $ \psi\left(\bar R^T_M\nabla_{\bar R_M} V\right)^T\Bigg]^T \in \mathbb{R}^{3M}$}, where $\nabla_{\bar R_k} V$ is the gradients of $ V$ with respect to $\bar R_k$ for all $k \in \mathcal{M}$. The time-derivative of $ V(x)$, along the trajectories of the closed-loop system (\ref{continuous_sys}), is given by
\begin{align}
    \dot{V}(x)&=-k_R\sum_{k=1}^{M} \text{tr}\left(A\bar{R}_k[\bar \sigma_k]^\times\right)\nonumber\\
    &=2 k_R\sum_{k=1}^{M}\bar \sigma_k^T \psi(A\bar{R}_k).\nonumber
\end{align}
The last equality was obtained using the facts that $\text{tr}\left(B[x]^\times\right)=\text{tr}\left(\mathbb{P}_a(B)[x]^\times\right)$ and $\text{tr}\left([x]^\times [y]^\times\right)=-2x^Ty$, $\forall x, y\in \mathbb{R}^3$ and $\forall B \in \mathbb{R}^{3\times3}$. In view of (\ref{s_bar}), one obtains
\begin{align}
    \dot{V}(x)&=2 k_R\bar{\sigma}^T \Psi = 2 k_R\sigma^T \textit{\textbf{H}} \Psi.
\end{align}
Furthermore, since $\sigma = -\textit{\textbf{H}} \Psi$, one has
\begin{align}\label{V_neg}
    \dot{V}(x)=-2 k_R ||\textit{\textbf{H}}\Psi||^2 \leq 0.
\end{align}
Thus, the desired equilibrium set $\mathcal{A}$ for system (\ref{continuous_sys}) is stable. Moreover, since the closed-loop system (\ref{continuous_sys}) is autonomous, as per LaSalle’s invariance theorem, any solution $x$ to the closed-loop system (\ref{continuous_sys}) must converge to the largest invariant set contained in the set characterized by $\dot{V}(x)=0$, \ie, $\textit{\textbf{H}}\Psi=0$. According to Lemma \ref{lem_H}, $\textit{\textbf{H}}\Psi=0$ implies $\Psi=0$. This also implies that
\begin{align}\label{eqq}
    A \bar R_k = \bar R_k^T A,
\end{align}
for every $k \in \mathcal{M}$. Thus, every solution $x$ of system (\ref{continuous_sys}) must converge to the set $\Upsilon$. This completes the proof of item (\ref{set_of_equilibrium}). Moreover, since $A$ is a real symmetric matrix, one can decompose $A$ as $A=U \Lambda U^T$ where $\Lambda=\text{diag}(\lambda_1, \lambda_2, \lambda_3)$ with $\lambda_1, \lambda_2$ and $\lambda_3$ are the distinct eigenvalues of $A$ and $U \in O(3)$. Using steps similar to \cite{mahony_tac2008} along with the fact that $A=U \Lambda U^T$, one can show that equation \eqref{eqq} further implies that $\bar R_k \in \{I_3, U D_1 U^T, U D_2 U^T, U D_3 U^T\}$, for every $k \in \mathcal{M}$, where $D_1=\text{diag}(1,-1,-1)$, $D_2=\text{diag}(-1,1,-1)$, $D_3=\text{diag}(-1,-1,1)$. It follows that the set of all undesired equilibrium points can be explicitly written as follows $\Upsilon \setminus \mathcal{A}=\{x \in \mathcal{S}: \bar R_m=I_3, \bar R_n=UD_{\hspace{-0.05cm}\beta}~U^T, ~m\in \mathcal{M}^{I},~ n \in \mathcal{M}^\pi, ~\beta\in\{1,2,3\}\}$ where $\mathcal{M}^I \cup \mathcal{M}^\pi=\mathcal{M}$, $|\mathcal{M}^\pi|>0$, $|\mathcal{M}^I|\geq0$.\\
Now, we will establish the stability properties of each equilibrium set. We start with the desired equilibrium set $\mathcal{A}$, and we set $\Tilde{R}_i=R_c \exp{\left([\tilde r^s_i]^\times\right)}$,  where $\tilde r^s_i\in \mathbb{R}^3$ is sufficiently small and $R_c \in SO(3)$ is an arbitrary constant rotation matrix. Considering the later expression of $\tilde R_i$ together with the fact that $\exp{\left([y]^\times\right)}\approx I_3+[y]^\times$, for sufficiently small $y$, one can get the following first-order approximation of $\tilde R_i$ around the desired equilibrium set $\mathcal{A}$:
\begin{equation}\label{first_order_app}
    \tilde R_i \approx R_c \left(I_3+[\tilde r_i^s]^\times \right),
\end{equation}
for every $i \in \mathcal{V}$. Moreover, it follows from \eqref{first_order_app}, with the fact that $A$ is symmetric, that
\begin{align}
    \mathbb{P}_a(A\tilde R_j^T \tilde R_i) \approx& \frac{1}{2}\bigg(A\left(I_3-[\tilde r_j]^\times\right) R_c^T R_c \left(I_3 + [\tilde r_i]^\times\right)\nonumber\\
    &-\left(I_3-[\tilde r_i]^\times\right) R_c^T R_c \left(I_3+[\tilde r_j]^\times\right) A\bigg).
\end{align}
Since we are only interested in the first-order approximation of the estimated attitude errors, the last equation can be simplified by neglecting the cross terms as follows:
{\small
\begin{align}
    \mathbb{P}_a(A\tilde R_j^T \tilde R_i) \approx \frac{1}{2}\bigg(\left(A[\tilde r_i]^\times + [\tilde r_i]^\times A\right)-\left(A[\tilde r_j]^\times+[\tilde r_j]^\times A\right)\bigg).
\end{align}}Furthermore, using the fact that $B [z]^\times + [z]^\times B^T = \left[\left(tr(B)I_3-B^T\right) z\right]^\times$, $ \forall z\in \mathbb{R}^3$, $\forall B \in \mathbb{R}^{3\times3}$, one has
\begin{equation}\label{p_lin}
    \mathbb{P}_a(A\tilde R_j^T \tilde R_i) \approx \frac{1}{2}\left[\bar A (\tilde r_i - \tilde r_j)\right]^\times,
\end{equation}
where $\Bar{A}:= \text{tr}(A)I_3-A$. From \eqref{first_order_app} and \eqref{p_lin}, one can derive the follwing linearization of \eqref{absolute_attitude_error}:
\begin{equation}
    R_c [\dot{\tilde r}_i]^\times = - \frac{k_R}{2} R_c \sum_{j\in \mathcal{N}_i}\left[\bar A (\tilde r_i - \tilde r_j)\right]^\times.
\end{equation}
where $i\in \mathcal{V}$. After some trivial mathematical manipulations, the following dynamics of $\tilde r^s_i$ are obtained:
\begin{equation}\label{cons_dyn}
    \Dot{\tilde r}^s_i=-\frac{k_R}{2} \Bar{A} \sum_{j\in \mathcal{N}_i} (\tilde r^s_i - \tilde r^s_j).
\end{equation}
 Equation \eqref{cons_dyn} represents the classical consensus protocol for multi-agent systems \cite{Ren_book, Mesbahi_book}. Note that, at the equilibrium point of system \eqref{cons_dyn} (\ie, $\tilde r_i^s = \tilde r_j^s$, $\forall i, j \in \mathcal{V}$), one has $\bar R_k = I_3$, for all $k \in \mathcal{M}$, which in turns implies that  $(\bar R_1, \bar R_2, \hdots, \bar R_M) \in \mathcal{A}$. Therefore, to show local asymptotic stability of the desired equilibrium set $\mathcal{A}$ , one has to show that the equilibrium point $\tilde r_i^s = \tilde r_j^s$, for all $i \in \mathcal{V}$ and $j\in \mathcal{N}_i$, of the system \eqref{cons_dyn}, is asymptotically stable. Defining $\tilde r^s:= \left[(\tilde r^s_1)^T, (\tilde r^s_2)^T, \hdots, (\tilde r^s_N)^T \right]^T$, it follows from \eqref{cons_dyn} that
\begin{align}\label{cons_equ}
    \Dot{\tilde r}^s&=-\frac{k_R}{2} \left(I_N \otimes \Bar{A}\right) \left( L\otimes I_3\right) \tilde r^s\nonumber\\
    &=-\frac{k_R}{2}\left( L\otimes \Bar{A}\right) \tilde r^s,
\end{align}
where $L= H H^T \in \mathbb{R}^{N\times N}$ is the Laplacian matrix corresponding to the graph $\mathcal{G}$. Since  $\mathcal{G}$ is undirected and connected (as per Assumption \ref{graph_ass}) and the matrix $A$ is positive definite with three distinct eigenvalues, it follows that the equilibrium point $\tilde r^s = \textbf{1}_N \otimes r_c$, for the multi-agent system \eqref{cons_equ}, is asymptotically stable, where $r_c = \frac{1}{N} \sum_{j=1}^N \tilde r^s_j(0)$. Consequently, the set $\mathcal{A}$ is locally
asymptotically stable. This completes the proof of item (\ref{A_local_stability}).\\
%Before proceeding with the proof of item \eqref{stability_of_equilibrium}, let first analysing the nature of the points belonging to the undesired equilibrium set $\Upsilon \setminus \mathcal{A}$ by finding the \textit{Hessian} matrix of $V(x)$ at these points and calculating its eigenvalues. 
To prove item \eqref{stability_of_equilibrium}, we first evaluate the \textit{Hessian} of $V(x)$, denoted by $\text{\textit{Hess}}V(x)$, to determine the nature of the critical points of $V(x)$ that belong to the set $\Upsilon \setminus \mathcal{A}$. Given an open interval $\mathbb{O} \subset \mathbb{R}$ containing zero in its interior, $\forall k\in \mathcal{M}$, one defines a smooth curve $\varphi_k: \mathbb{O} \rightarrow SO(3)$ such that $\varphi_k(t)=\Bar{R}^*_k \exp{\left(t [\zeta_k]^\times\right)}$ where $\zeta_k \in \mathbb{R}^3$ and $x^*=(\bar R^*_1, \bar R^*_2, \hdots, \bar R^*_M) \in \Upsilon \setminus \mathcal{A}$. Let $x_\varphi (t):= \left(\varphi_1(t), \varphi_2(t), \hdots, \varphi_M(t)\right) \in \mathcal{S}$, one has
\begin{align}\label{hess}
    \frac{d}{dt}V(x_\varphi)=&-\sum_{k=1}^{M} \text{tr}\left(A\bar{R}^*_k\exp{\left(t [\zeta_k]^\times\right)}[\zeta_k]^\times\right)\nonumber\\
    \frac{d^2}{dt^2}V(x_\varphi)=&-\sum_{k=1}^{M} \text{tr} \left(A\bar{R}^*_k\exp{\left(t [\zeta_k]^\times\right)}\left([\zeta_k]^\times\right)^2\right)\nonumber\\
    &-\sum_{k=1}^{M} \text{tr}\left(A\bar{R}^*_k\exp{\left(t [\zeta_k]^\times\right)}[\dot \zeta_k]^\times\right).
\end{align}
Since $x_\varphi(0)=x^*$, one verifies $\mathbb{P}_a(A\bar R^*_k)=0$ for every $k \in \mathcal{M}$. Consequently, it follows from (\ref{hess}) that 
\begin{align}
    \left.\frac{d^2}{dt^2}V(x_\varphi)\right|_{t=0}&=-\sum_{k=1}^{M} \text{tr}\left(A\bar R_k^* \left([\zeta_k]^\times\right)^2\right).
\end{align}
Using the fact $\left( [z]^\times\right)^2=-z^T z I_3 + zz^T$ and $\text{tr}(z_1 z_2^T)= z_1^T z_2$, $\forall z, z_1, z_2 \in \mathbb{R}^3$, one obtains 
\begin{align}\label{hess2}
    \left.\frac{d^2}{dt^2}V(x_\varphi)\right|_{t=0}&=\sum_{k=1}^{M} \zeta^T_k \left( \text{tr}(A\bar R^*_k)I_3 - A \bar R^*_k\right) \zeta_k\nonumber\\
    &= \sum_{k=1}^{M} \zeta^T_k A^*_k \zeta_k=\zeta^T \textit{\textbf{A}}^* \zeta,
\end{align}
where $A^*_k=\text{tr}(A\bar R^*_k)I_3 - A \bar R^*_k$, $\zeta = [\zeta_1^T, \zeta_2^T, \hdots, \zeta_M^T]^T \in \mathbb{R}^{3M}$ and $\textit{\textbf{A}}^* = \text{diag}(A^*_1, A^*_2, \hdots, A^*_M)\in \mathbb{R}^{3M\times 3M}$. In view of (\ref{hess2}), according to \cite{Mahony_book_OAMM}, one has $\textit{Hess}V(x)=\textit{\textbf{A}}^*$ for every $x \in \Upsilon\setminus \mathcal{A}$. 
In other words, the matrix $\textit{\textbf{A}}^*$ represents the \textit{Hessian} of $V(x)$ evaluated at the undesired equilibrium points. It is worth noting that the eigenvalues of the matrix $\textit{\textbf{A}}^*$ are actually the eigenvalues of the matrices $A^*_k$, for every $k \in \mathcal{M}$. Therefore, as a next step, we will explicitly find the eigenvalues of the matrices $A^*_k$, for every $k \in \mathcal{M}$. Using the fact that $A=U\Lambda U^T$, recall that $U^T U=I_3$ and $\Lambda = \text{diag}(\lambda_1, \lambda_2, \lambda_3)$ with $\lambda_1\neq \lambda_2\neq\lambda_3$ (according to Assumption \ref{A_ass}), one has $A^*_k =U\left(\text{tr}(\Lambda U^T \bar R^*_k U)I_3-\Lambda U^T \bar R^*_k U\right)U^T$. Now, for every $m \in \mathcal{M}^I$, one can verify that $\text{tr}(\Lambda U^T \bar R^*_m U)I_3-\Lambda U^T \bar R^*_m U = \text{diag}(\lambda_2+\lambda_3, \lambda_1+\lambda_3, \lambda_1+\lambda_2)$. On the other hand, for every $n \in \mathcal{M}^\pi$, one can verify that $\text{tr}(\Lambda U^T \bar R^*_n U)I_3-\Lambda U^T \bar R^*_n U \in \{\text{diag}(-\lambda_2-\lambda_3, \lambda_1-\lambda_3, \lambda_1-\lambda_2), \text{diag}(\lambda_2-\lambda_3, -\lambda_1-\lambda_3, \lambda_2-\lambda_1), \text{diag}(\lambda_3-\lambda_2, \lambda_3-\lambda_1, -\lambda_1-\lambda_2)\}$. Since $\lambda_1\neq \lambda_2\neq\lambda_3$, it follows that the eigenvalues of the matrix $\textit{\textbf{A}}^*$ are either all negative or some of them are positive and some are negative. Consequently, the critical points of $V(x)$ in $\Upsilon \setminus \mathcal{A}$ are either global maxima or saddle points of $V(x)$.\\
Now, we will show that the critical points of $V(x)$ in the set $\Upsilon \setminus \mathcal{A}$ are unstable.
%Without loss of generality, for every $x^* \in \Upsilon \setminus \mathcal{A}$, let us consider $\bar R_n = \mathcal{R}_\alpha(\pi,v^n_l)$, $n \in \mathcal{M}^\pi$, where $v_l^n$ is the eigenvector of the matrix $A$ corresponding to the eigenvalue $\lambda_l^n$. 
Consider the following real-valued function $\bar V: SO(3)^M \rightarrow \mathbb{R}$ inspired from \cite{vantran2019pose_arXiv}:
%\begin{align}\label{V_bar}
%    \bar V(x)=& \sum_{n \in \mathcal{M}^\pi} \left(2(\lambda^n_p+\lambda^n_d)-\text{tr}\left(A(I_3-\bar R_n)\right)\right)\nonumber\\
%    &-\sum_{m \in \mathcal{M}^I} \text{tr}\left(A(I_3-\bar R_m)\right),
%\end{align}
%where $\lambda^n_p$ and $\lambda^n_d$ are the eigenvalues of the matrix $A$. Recall that, according to Assumption \ref{A_ass}, one has $\lambda^n_p\neq\lambda^n_d\neq\lambda^n_l$, for every $n \in \mathcal{M}^\pi$. 
%It follows from \eqref{V_bar} that  
%\begin{align}
%    \bar V(x)=& 2\sum_{n \in \mathcal{M}^\pi}(\lambda^n_p+\lambda^n_d)-\sum_{n \in \mathcal{M}^\pi}\text{tr}\left(A(I_3-\bar R_n)\right)\nonumber\\
%    &-\sum_{m \in \mathcal{M}^I} \text{tr}\left(A(I_3-\bar R_m)\right).
%\end{align}
%Since $\mathcal{M}=\mathcal{M}^I \cup \mathcal{M}^\pi$, one obtains 
\begin{align}\label{V_bar_1}
    \bar V(x)=& 2\sum_{n \in \mathcal{M}^\pi}(\lambda_{p_n}+\lambda_{d_n})- V(x).
\end{align}
 %2\sum_{n \in \mathcal{M}^\pi}(\lambda^n_p+\lambda^n_d)-\sum_{k \in \mathcal{M}}\text{tr}\left(A(I_3-\bar R_k)\right)\nonumber\\
where $\lambda_{p_n}$ and $\lambda_{d_n}$ are two distinct eigenvalues of $A$, \ie, $p_n, d_n \in \{1, 2, 3\}$ such that $p_n \neq d_n$. 
%Recall that, according to Assumption \ref{A_ass}, one has $\lambda^n_p\neq\lambda^n_d\neq\lambda^n_l$, for every $n \in \mathcal{M}^\pi$. 
Let us consider an equilibrium point $x^* \in \Upsilon \setminus \mathcal{A}$ such that $\bar R_n = U D_{l_n} U^T$, $n \in \mathcal{M}^\pi$, where $l_n \in \{1, 2, 3\}$ such that $l_n\not=p_n$ and $l_n\not=d_n$.
It is clear that $\bar V(x^*)=0$. Now, define the following set $\mathbb{B}_r:=\{ (\bar R_1, \bar R_2, \hdots, \bar R_M) \in \mathcal{S}:~|\bar R_1^T \bar R_1^*|_I+|\bar R_2^T \bar R_2^*|_I+ \hdots+ |\bar R_M^T \bar R_M^*|_I \leq r\}$ where $r >0$. Since the set $\Upsilon \setminus \mathcal{A}$ contains only global maxima or saddle points of $V(x)$, one can verify that the set $\mathbb{U}=\{x \in \mathbb{B}_r:~\bar V(x)>0\}$ is not empty. Furthermore, it follows from \eqref{V_neg} and \eqref{V_bar_1} that $\dot{\bar V}(x)=-\dot{V}(x)>0$. Consequently, any trajectory that initially starts in the set $U$ must leave the set $U$. Thus, according to \textit{Chetaev's} theorem \cite{khalil2002nonlinear}, one concludes that all points that belong to the undesired equilibrium set $\Upsilon \setminus \mathcal{A}$ are unstable. By virtue of the stable manifold theorem \cite{Perko_book}, one can conclude that the stable manifold associated to the undesired equilibrium set $\Upsilon \setminus \mathcal{A}$ has zero Lebesgue measure, and as such, the equilibrium set $\mathcal{A}$ is almost globally asymptotically stable. This completes the proof of item (\ref{stability_of_equilibrium}).
\end{proof}
\begin{rmk}
     In contrast to \cite{lee_2016, lee_auto2016, Tran_CDC2018, lee_2019,li_2020}, our proposed continuous attitude observer can estimate time-varying orientations. Moreover, the estimated orientations provided by our proposed scheme are well-defined for any instant of time, which is not the case in \cite{Tran_TCNS2019,lee_2019,li_2020}, since the reliable attitude estimates are obtained only at the steady state. This makes our proposed distributed attitude estimation scheme a strong candidate for use in applications that require instantaneous orientations for feedback.
\end{rmk}

\begin{rmk}\label{R1}
As shown in the proof of Theorem \ref{theorem1}, the trajectories of the closed-loop system (\ref{continuous_sys}) may converge to a level set containing the undesired equilibrium set $\Upsilon \setminus \mathcal{A}$. Unfortunately, due to the topological obstruction on $SO(3)$, there is no continuous distributed attitude estimation scheme that guarantees global asymptotic stability of the desired equilibrium set $\mathcal{A}$ \cite{Koditschek}. Therefore, in the sequel, we will propose a hybrid distributed attitude estimation scheme endowed with global asymptotic stability.
\end{rmk}

\subsection{Hybrid Distributed Attitude Estimation}
The design of the correcting term \eqref{sigma_c} was based on the gradient of the smooth potential function $ V(x)$. However, this design does not guarantee the global stability of the desired equilibrium set $\mathcal{A}$. This is due to the fact that the potential function $ V(x)$ has more than one critical set ($\mathcal{A}$ and $\Upsilon\setminus\mathcal{A}$). It is well known that the design of gradient-based laws using smooth potential functions on $SO(3)$ leads to the problem mentioned above \cite{Morse_book}. This has motivated many authors to propose hybrid gradient-based solutions that ensure the existence of a unique global attractor \cite{Mayhew_ACC2011, Mayhew_ACC2011_2, Mayhew_TAC2013, 7462234}. The key idea in these solutions is to switch between a family of smooth potential functions via an appropriate switching mechanism, which leads to generating a non-smooth gradient with only one global attractor. Note that the construction of this family of smooth potential functions relies on the compactness assumption of the manifold. Therefore, the hybrid gradient-based solutions proposed in \cite{Mayhew_ACC2011, Mayhew_ACC2011_2, Mayhew_TAC2013, 7462234} are not applicable to non-compact manifolds such as $SE(3)$. Recently, the authors of \cite{miaomiao_TAC2022} proposed a new hybrid scheme that relies only on one potential function parameterized by a scalar variable governed by hybrid dynamics. According to an appropriate switching mechanism applied to this variable, the potential function can be adjusted so that the generated non-smooth gradient has only one global attractor. In contrast to \cite{Mayhew_ACC2011, Mayhew_ACC2011_2, Mayhew_TAC2013, 7462234}, the hybrid scheme, given in \cite{miaomiao_TAC2022}, is easy to design and does not require any assumption about the compactness of the manifold.\\
Inspired by \cite{miaomiao_TAC2022},  we introduce a new potential function $U_R$, on $SO(3)^M\times\mathbb{R}^M$, by augmenting the argument of $ V$ with some auxiliary time-varying scalar variables. Designing an appropriate switching mechanism to govern these auxiliary variables will allow the generation of non-smooth gradient-based laws with only one global attractor.\\
%The design of the correcting term \eqref{sigma_c} was based on the gradient of the smooth potential function $ V(x)$. However, this design does not guarantee the global stability of the desired equilibrium set $\mathcal{A}$. This is due to the fact that the potential function $ V(x)$ has more than one critical set ($\mathcal{A}$ and $\Upsilon\setminus\mathcal{A}$). It is well known that the design of gradient-based laws using smooth potential functions on $SO(3)$ leads to the problem mentioned above \cite{Morse_book}. To address this challenge, inspired by \cite{miaomiao_TAC2022},  we introduce a new potential function $U_R$, on $SO(3)^M\times\mathbb{R}^M$, by augmenting the argument of $ V$ with some auxiliary time-varying scalar variables. Designing an appropriate switching mechanism to govern these auxiliary variables will allow the generation of non-smooth gradient-based laws with only one global attractor.\\
\subsubsection{Switching Mechanism Design for Multi-agent Systems}
\indent Let $\mathcal{A}_h:=\{x_h\in \mathcal{S}_h: \forall k\in \mathcal{M}, \hspace{0.1cm} \bar{R}_k=I_3, \xi_k=0\}$, where $x_h:= \left(\bar R_1, \hdots, \bar R_M, \xi_1, \hdots, \xi_M\right)\in \mathcal{S}_h$ with $\mathcal{S}_h:=SO(3)^M\times\mathbb{R}^M$. Consider the following potential function, on $\mathcal{S}_h$, with respect to $\mathcal{A}_h$:
\begin{align}\label{potential_fct}
    U_R(x_h) &= \sum_{k=1}^{M}U(\bar R_k,\xi_k),
\end{align} 
where $U: SO(3)\times \mathbb{R} \rightarrow \mathbb{R}_{\geq0}$ is a potential function with respect to $(I_3, 0)$. The following set represents the set of all critical points of $U_R$:
\begin{equation}\label{equi_set_}
    \Upsilon_h := \{x_h \in \mathcal{S}_h: \forall k\in \mathcal{M}, \nabla_{\bar R_k}U_R=0~\text{and}~\nabla_{\xi_k}U_R=0\},
\end{equation}
where $\nabla_{\bar R_k}U_R$ and $\nabla_{\xi_k}U_R$ are the gradients of $U_R$ with respect to $\bar R_k$ and $\xi_k$, respectively. The potential function $U_R$ is chosen such that $\mathcal{A}_h \subset \Upsilon_h$. In what follows, inspired by \cite{miaomiao_TAC2022}, we will introduce an essential condition for our hybrid scheme design related to the potential function $U_R$.
\begin{cnd}\label{hybrid_ass}
    Consider the potential function \eqref{potential_fct}. There exist a scalar $\delta>0$ and a nonempty finite set $\Xi \subset \mathbb{R}$ such that for every $x_h \in \Upsilon_h \setminus \mathcal{A}_h$
     \begin{equation}\label{cond1}
        U(\bar R_k, \xi_k) - \underset{\bar{\xi}_k \in \Xi}{\text{min}}~U(\bar R_k, \bar{\xi}_k) > \delta,
    \end{equation}
   for every $k \in \mathcal{M}$ such that $\bar R_k \neq I_3$. 
\end{cnd}
\begin{rmk}
    Note that the set $\Upsilon_h\setminus \mathcal{A}_h$ denotes the set of all undesired critical points of $U_R$. Condition \ref{hybrid_ass} plays a key role in the design of the hybrid scheme to be introduced shortly, since it implies that, at the undesired critical set $\Upsilon_h \setminus \mathcal{A}_h$, there will always exist $\bar \xi_k \in \Xi$, for every $k \in \mathcal{M}$, where $\bar R_k \neq I_3$, such that $U(\bar R_k, \bar{\xi}_k)$ is lower than $U(\bar R_k, \xi_k)$ by a constant gap $\delta$. Thus, resetting the value of $\xi_k$ to $\bar \xi_k$ will effectively steer the state away from the undesired critical set $\Upsilon_h \setminus \mathcal{A}_h$. This, together with an appropriate design of the vector field that ensures that $U_R$ is non-increasing during the flows, will guarantee global asymptotic stability of the desired equilibrium set $\mathcal{A}_h$.
\end{rmk}
    Now, for every $i \in \mathcal{V}$ and $k \in \mathcal{M}_i^+$, we propose the following hybrid dynamics for $\xi_k$:
    \begin{align}
     &\underbrace{
         \begin{aligned}
             \dot{\xi}_k =&-k_\xi \nabla_{\xi_k}U_R
         \end{aligned}
        }_{x_h \in \mathcal{F}_i} \label{H_observer_dynamics_1}
        \\
     &\underbrace{        
        \xi_k^+\in 
        \begin{cases}
        \xi_k~\text{if}~~U(\bar{R}_k, \xi_k)-U(\bar{R}_k, \xi^*_k) \leq \delta\\
        \xi^*_k~\text{if}~~U(\bar{R}_k, \xi_k)-U(\bar{R}_k, \xi^*_k) \geq \delta,
        \end{cases}}_{x_h \in \mathcal{J}_i}\label{H_observer_dynamics_2}
\end{align}
where $k_\xi >0$ and $\xi^*_k := \text{arg} \underset{\bar{\xi}_k \in \Xi}{\text{min}} U(\bar{R}_k, \bar{\xi}_k)$. The flow set $\mathcal{F}_i$ and the jump set $\mathcal{J}_i$, for agent $i$, are defined as follows:
{\small
\begin{align}    
    \mathcal{F}_i&:=\{x_h\in \mathcal{S}_h : \forall k\in \mathcal{M}_i^+,\hspace{0cm} U(\bar{R}_k,\xi_k)-\underset{\bar{\xi}_k\in \Xi}{\text{min}} U(\bar{R}_k,\bar{\xi}_k)\leq \delta\},\nonumber\\
    \mathcal{J}_i&:=\{x_h \in \mathcal{S}_h: \exists k\in \mathcal{M}_i^+, 
    \hspace{0cm} U(\bar{R}_k,\xi_k)-\underset{\bar{\xi}_k\in \Xi}{\text{min}} U(\bar{R}_k,\bar{\xi}_k)\geq \delta\}.\nonumber
\end{align}}It is important to note that for each $i \in \mathcal{V}$, both the flow set $\mathcal{F}_i$ and the jump set $\mathcal{J}_i$ are distributed. In other words, for each $i \in \mathcal{V}$, the definition of both sets, namely $\mathcal{F}_i$ and $\mathcal{J}_i$, considers only the edges connecting agent $i$, where $i$ is the head of these oriented edges, with its neighbors $j \in \mathcal{N}_i$. Now, let $\xi:=[\xi_1, \xi_2, \hdots, \xi_M]^T \in \mathbb{R}^M$, from \eqref{H_observer_dynamics_1}-\eqref{H_observer_dynamics_2}, one obtains the following hybrid dynamics
\begin{equation}\label{h_xi}
\mathcal{H}_\xi:\begin{cases} {\dot \xi = F_\xi(x_h)}&{x_h \in {\mathcal{F}}} \\ {{\xi^ + } \in G_\xi(x_h)}&{x_h \in {\mathcal{J}}} \end{cases}
\end{equation}
where
\begin{eqnarray}\label{network_f_j_set}
\mathcal{F}:=\bigcap_{i=1}^{N}\mathcal{F}_{i},\qquad \mathcal{J}:=\bigcup_{i=1}^{N}\mathcal{J}_{i},
\end{eqnarray}
and
\begin{equation}
    F_\xi(x_h) = \left[ {\begin{array}{c} {{-k_\xi \nabla_{\xi_1}U_R}} \\ \vdots\\ {-k_\xi \nabla_{\xi_M}U_R} \end{array}}\right],~G_\xi(x_h) = \left[ {\begin{array}{c} {{\{\xi_1, \xi_1^*\}}} \\ \vdots\\ {\{\xi_M, \xi_M^*\}} \end{array}}\right]\nonumber,
\end{equation}
\begin{rmk}
   Based on the definition of the jump set $\mathcal{J}$, one can verify that the set of all undesired critical points belongs to the jump set $\mathcal{J}$, \ie, $\Upsilon_h \setminus \mathcal{A}_h \subset \mathcal{J}$. Also, in view of the jump map $G_\xi(x_h)$, one can check that, if $x_h \in \mathcal{J}$, there exits $k \in \mathcal{M}$ such that $U(\bar{R}_k, \xi_k)-U(\bar{R}_k, \xi^*_k) \geq \delta$, which ensures that $U_R$ decreases at least by $\delta$ after each jump.\\ 
\end{rmk}

\subsubsection{Generic Hybrid Attitude Observer}
Consider the following hybrid attitude observer for every $i \in \mathcal{V}$:
\begin{align}
     \underbrace{
         \begin{aligned}
             \dot{\hat{R}}_i=\hat{R}_i[\omega_i-k_R\hat{R}_i^T \sigma_i]^\times
         \end{aligned}
        }_{x_h \in \mathcal{F}_i} %\label{H_observer_dynamics_1}
        ~~~~~~~~~~
     \underbrace{        
        \hat R_i^+=\hat R_i,
        }_{x_h \in \mathcal{J}_i}\label{TH_observer_dynamics_1}
\end{align}
with the following distributed gradient-based correcting term:
{\small
\begin{equation}\label{TH_observer_dynamics_2}
\sigma_i = \sum_{l\in \mathcal{M}_i^-} \bar R_l \psi \left(\bar R_l^T \nabla_{\bar R_l}U_R(x_h)\right)-\sum_{n\in \mathcal{M}_i^+} \psi \left(\bar R_n^T \nabla_{\bar R_n}U_R(x_h)\right)
\end{equation}}where the hybrid dynamics of $\xi_k$, $k \in \mathcal{M}_i^+$, are given in \eqref{H_observer_dynamics_2}. Moreover, given the above expression of $\sigma_i$, for every $i\in \mathcal{V}$, together with equation \eqref{s_bar}, one can verify that
    \begin{equation}\label{relative_correcting_term}
      \bar \sigma = - \textit{\textbf{H}}^T\textit{\textbf{H}}~\Psi_\nabla^{\bar R},
    \end{equation}
    where {\small $\Psi_\nabla^{\bar R} :=\Bigg[\psi\left(\bar R^T_1\nabla_{\bar R_1}U_R\right)^T,\psi\left(\bar R^T_2\nabla_{\bar R_2}U_R\right)^T, \hdots,$ $ \psi\left(\bar R^T_M\nabla_{\bar R_M}U_R\right)^T\Bigg]^T \in \mathbb{R}^{3M}$} and the block matrix $\textit{\textbf{H}}$ is given in (\ref{H_bar}). The matrix $\textit{\textbf{H}}^T \textit{\textbf{H}}$ is positive definite according to Lemma \ref{lem_H}. Next, we will establish the stability property of the desired equilibrium set $\mathcal{A}_h$ considering the distributed hybrid attitude observer \eqref{TH_observer_dynamics_1}-\eqref{TH_observer_dynamics_2}. In view of \eqref{absolute_attitude_error}-\eqref{s_bar}, \eqref{H_observer_dynamics_2} and \eqref{TH_observer_dynamics_1}-\eqref{relative_correcting_term}, one can derive the following multi-agent hybrid closed-loop dynamics:  
\begin{equation}\label{hybrid_sys}
\mathcal{H}:\begin{cases} {\dot x_h = F(x_h)}&{x_h \in {\mathcal{F}}} \\ {{x_h^ + } \in G(x_h)}&{x_h \in {\mathcal{J}}} \end{cases}
\end{equation}
where
\begin{equation}
    F(x_h) = \left[ {\begin{array}{c} {k_R\bar{R}_1[\bar\sigma_1]^\times} \\ \vdots\\ {k_R\bar{R}_M[\bar\sigma_M]^\times}\\
    {-k_\xi\nabla_{\xi_1}U_R}\\
    \vdots\\{-k_\xi\nabla_{\xi_M}U_R}
    \end{array}}\right],~G(x_h) = \left[ {\begin{array}{c} {\bar{R}_1} \\ \vdots\\ {\bar{R}_M} \\ \{  \xi_1, \xi^*_1\} \\ \vdots\\\{\xi_M, \xi^*_M\}  \end{array}}\right]\nonumber.
\end{equation}
The flow set $\mathcal{F}$ and the jump set $\mathcal{J}$ are defined in \eqref{network_f_j_set}. It is worth noting that the hybrid closed-loop system (\ref{hybrid_sys}) is autonomous.
\begin{rmk}
    According to the flow map $F(x_h)$, the dynamics of the state $x_h$ flow along a negative direction of the gradient of $U_R$, driving the state $x_h$ towards the critical points of $U_R$ during the flow. However, the jump map $G(x_h)$ pushes the state $x_h$ away from the undesired critical set $\Upsilon_h \setminus \mathcal{A}_h$, which leaves the desired critical set $\mathcal{A}_h$ as a global attractor to our proposed distributed hybrid correcting scheme.
\end{rmk}

Before proceeding with the stability analysis, it is important to verify that system \eqref{hybrid_sys} is well-posed\footnote{See \cite[Definition 6.2]{goebel2012hybrid} for the definition of well-posedness.}. This involves showing that the hybrid closed-loop system (\ref{hybrid_sys}) satisfies the hybrid basic conditions \cite[Assumption 6.5]{goebel2012hybrid}. 
\begin{lem}\label{lem_hbc}
    The hybrid closed-loop system (\ref{hybrid_sys}) satisfies the following hybrid basic conditions:
    \begin{enumerate}[i)]
        \item $\mathcal{F}$ and $\mathcal{J}$ are closed subsets;\label{hbc_1}
        \item $F$ is outer semicontinuous and locally bounded relative to $\mathcal{F}$,
           $\mathcal{F} \in \text{dom}~F$ , and $F(x_h)$ is convex for every $x_h \in \mathcal{F}$;\label{hbc_2}
        \item $G$ is outer semicontinuous and locally bounded relative to $\mathcal{J}$
           and $\mathcal{J} \in \text{dom}~G$.\label{hbc_3}  
    \end{enumerate}
\end{lem}
\begin{proof}
    Since the function $U$ is continuous, one can verify that the flow set $\mathcal{F}$ and the jump set $\mathcal{J}$, given in \eqref{network_f_j_set}, are closed sets. Moreover, one has $\mathcal{F}\cup \mathcal{J} = \mathcal{S}_h$. \\
    The outer semicontinuity, local boundedness, and convexity properties of the flow map $F$ follow from the fact that $F$ is a single-valued continuous function.\\
    Using the fact that $U$ is continuous on $SO(3) \times \mathbb{R}$, one can show, following similar arguments as in \cite[Proof of Lemma 1]{Casau_TAC2019}, that $\rho_k(\bar R_k):= \text{arg} \min_{\bar{\xi}_k \in \Xi} U(\bar{R}_k, \bar{\xi}_k)$ for every $k \in \mathcal{M}$ is outer semicontinuous. Furthermore, it can be verified that for every $i \in \mathcal{V}$ and $k \in \mathcal{M}_i^+$, the set-valued mapping given in \eqref{H_observer_dynamics_2} has a closed graph relative to $\mathcal{J}_i$. According to \cite[Lemma 5.10]{goebel2012hybrid}, this implies that the set-valued mapping \eqref{H_observer_dynamics_2} is outer semi-continuous relative to $\mathcal{J}_i$. Consequently, in conjunction with the fact that $\rho_k(\bar R_k)$ is outer semi-continuous, it follows that the jump map $G$ is outer semi-continuous relative to $\mathcal{J}$. The local boundedness of $G$ relative to $\mathcal{J}$ follows from the fact that $\xi_k^*$, for every $k\in \mathcal{M}$, takes values over a finite discrete set $\Xi$ and the remaining components of $G$ are single-valued continuous functions on $\mathcal{J}$.
\end{proof}
In the following theorem, we will establish the stability properties of the multi-agent hybrid closed-loop system \eqref{hybrid_sys}.
\begin{thm}\label{theorem2}
   Let $k_R, k_\xi >0$ and suppose that Assumptions \ref{hybrid_ass}, \ref{graph_ass} and Condition \ref{hybrid_ass} hold. Then, the set $\mathcal{A}_h$ is globally asymptotically stable for the multi-agent hybrid closed-loop system (\ref{hybrid_sys}) and the number of jumps is finite. 
\end{thm}
\begin{proof}
    Consider the following Lyapunov function candidate:
\begin{align}\label{l_fct}
    U_R(x_h) = \sum_{k=1}^{M}U(\bar R_k,\xi_k).
\end{align}
We take the potential function $U_R$ as a Lyapunov function candidate to study the stability property of the set $\mathcal{A}_h$ for the hybrid closed-loop system (\ref{hybrid_sys}). The time-derivative of (\ref{l_fct}), along the trajectories generated by the flows of the hybrid closed-loop system (\ref{hybrid_sys}), is given by
\begin{align}\label{U_grd}
        \dot U_R(x_h) =& \sum_{k=1}^{M}\langle \nabla_{\bar R_k}U_R, k_R \bar R_k [\bar \sigma_k]^\times\rangle_{\bar R_k}+\sum_{k=1}^{M}\langle \langle \nabla_{\xi_k}U_R, \dot \xi_k \rangle \rangle \nonumber\\
        =& \sum_{k=1}^{M}\langle \langle \bar R_k^T \nabla_{\bar R_k}U_R, k_R[\bar \sigma_k]^\times\rangle \rangle+\sum_{k=1}^{M}\langle \langle \nabla_{\xi_k}U_R, \dot \xi_k \rangle \rangle \nonumber\\
        =& 2 k_R\sum_{k=1}^{M} \bar \sigma_k^T \psi\left(\bar R^T_k\nabla_{\bar R_k}U_R\right)+\sum_{k=1}^{M}\dot \xi_k\nabla_{\xi_k}U_R\\
        =& 2 k_R \bar \sigma^T \Psi_\nabla^{\bar R}+\dot \xi^T \Psi_\nabla^\xi,\nonumber
    \end{align}
    where $\Psi_\nabla^\xi :=\left[\nabla_{\xi_1}U_R,\nabla_{\xi_2}U_R, \hdots, \nabla_{\xi_M}U_R\right]^T \in \mathbb{R}^M$. To derive the above equations, the following facts have been used: $\langle \eta_1, \eta_2 \rangle_R=\langle\langle R^T\eta_1, R^T\eta_2 \rangle\rangle$, $\text{tr}\left(B[x]^\times\right)=\text{tr}\left(\mathbb{P}_a(B)[x]^\times\right)$ and $\text{tr}\left([x]^\times [y]^\times\right)=-2x^Ty$, $\forall x, y\in \mathbb{R}^3$, $\forall B \in \mathbb{R}^{3\times3}$, $\forall \eta_1, \eta_2 \in \mathfrak{so}(3)$ and $\forall R \in SO(3)$. It follows from (\ref{H_observer_dynamics_1}) and (\ref{relative_correcting_term}) that 
    \begin{equation}\label{u_dot}
        \dot U_R(x_h) = -2k_R||\textit{\textbf{H}} \Psi_\nabla^{\bar R}||^2-k_\xi || \Psi_\nabla^\xi||^2 \leq 0.
    \end{equation}
    Consequently, $U_R(x_h)$ is non-increasing along the flows of (\ref{hybrid_sys}). Moreover, in view of (\ref{hybrid_sys}) and (\ref{potential_fct}), one has
\begin{align}\label{u_+}
    U_R(x_h)-U_R(x_h^+)=&\sum_{k=1}^{M}\left(U(\bar{R}_k, \xi_k)-U(\bar{R}_k^+, \xi_k^+)\right)\nonumber\\
    \geq&\delta.
\end{align}
Thus, $U_R(x_h)$ is strictly decreasing over the jumps of \eqref{hybrid_sys}. It follows from \eqref{u_dot}-\eqref{u_+} and the result presented in \cite[Theorem 23]{Goebel_ieee_magazine} that the set $\mathcal{A}_h$ is stable. Consequently, every maximal solution of the hybrid closed-loop system (\ref{hybrid_sys}) is bounded. In addition, from \eqref{u_dot} and \eqref{u_+}, one can verify that $U_R(x_h(t,j))\leq U_R(x_h(t_j,j))$ and $U_R(x_h(t_j,j))\leq U_R(x_h(t_j,j-1))-\delta$, $\forall (t,j), (t_j,j), (t_j,j-1) \in \text{dom} \hspace{0.1cm} x_h$, with $(t,j)\geq (t_j,j) \geq (t_j,j-1)$. Thus, one has $0\leq U_R(x_h(t,j))\leq U_R(x_h(0,0))-j\delta$, $\forall (t,j)\in \text{dom} \hspace{0.1cm} x_h$, which leads to $j\leq \lceil \frac{U_R(x_h(0,0))}{\delta}\rceil$, where $\lceil . \rceil$ denotes the ceiling function. This shows that the number of jumps is finite and depends on the initial conditions.\\
Now, we will show the global attractivity of $\mathcal{A}_h$ using the invariance principle for hybrid systems \cite[Section 8.2]{goebel2012hybrid}. Consider the following functions:
\begin{align}
     u_{\mathcal{F}}(x_h) &:=\begin{cases}
                    -2k_R||\textit{\textbf{H}} \Psi_\nabla^{\bar R}||^2-k_\xi || \Psi_\nabla^\xi||^2~~~~~ \text{if}~x_h \in \mathcal{F},\\
                     -\infty~~~~~~~~~~~~ \text{otherwise},
                \end{cases}\label{uf}\\
     u_{\mathcal{J}}(x_h) &:=\begin{cases}
                    &-\delta~~~~~~~~~~~~ \text{if}~x_h \in \mathcal{J},\\
                    & -\infty~~~~~~~~~~ \text{otherwise},
                \end{cases}\label{uj}
\end{align}
In view of \eqref{u_dot}-\eqref{uj}, one can notice that the growth of $U_R$ is upper bounded during the flows by $u_{\mathcal{F}}(x_h) \leq 0$ and during the jumps by $u_{\mathcal{J}}(x_h) \leq 0$ for every $x_h \in \mathcal{S}_h$. 
It follows from \cite[Corollary 8.4]{goebel2012hybrid} that every maximal solution of the hybrid system \eqref{hybrid_sys} converges to the following largest weakly\footnote{The reader is referred to \cite{goebel2012hybrid} for the definition of \textit{weakly invariant} sets in the hybrid systems context.} invariant subset: $$U_R^{-1}(r)\cap \mathcal{S}_h \cap \left[\overline{u_{\mathcal{F}}^{-1}(0)} \cup \left(u_{\mathcal{J}}^{-1}(0)\cap G \left(u_{\mathcal{J}}^{-1}(0)\right)\right)\right],$$ for some $r \in \mathbb{R}$, where $\overline{u_{\mathcal{F}}^{-1}(0)}$ denotes the closure of the set $u_{\mathcal{F}}^{-1}(0)$. Moreover, one can verify that 
\begin{align}
    u_{\mathcal{F}}^{-1}(0)&=\{x_h \in \mathcal{F}:~\textit{\textbf{H}} \Psi_\nabla^{\bar R}=0,~ \Psi_\nabla^\xi=0\} \nonumber\\
    u_{\mathcal{J}}^{-1}(0)&= \emptyset. \nonumber
\end{align}
Furthermore, according to Lemma \ref{lem_H}, one has
\begin{align}
    u_{\mathcal{F}}^{-1}(0)&=\{x_h \in \mathcal{F}:~\Psi_\nabla^{\bar R}=0,~ \Psi_\nabla^\xi=0\} \nonumber\\
    &= \mathcal{F} \cap \Upsilon_h, \nonumber
\end{align}
where the set $\Upsilon_h$ is defined in \eqref{equi_set_}. Given $x_h \in \mathcal{A}_h$, one obtains, for all $k \in \mathcal{M}$, $U(\bar{R}_k, \xi_k)-\underset{\bar{\xi}_k\in \Xi}{\text{min}} U(\bar{R}_k,\bar{\xi}_k)=-\underset{\bar{\xi}_k\in \Xi}{\text{min}} U(\bar{R}_k,\bar{\xi}_k)\leq 0$. Therefore, from (\ref{network_f_j_set}), and according to Condition \ref{hybrid_ass}, one can verify that $\mathcal{A}_h \subset \mathcal{F}\cap \Upsilon_h$ and $\mathcal{F}\cap (\Upsilon_h \setminus \mathcal{A}_h)=\emptyset$. In addition, applying some set-theoretic arguments, one has $\mathcal{F} \cap \Upsilon_h \subset (\mathcal{F} \cap (\Upsilon_h\setminus \mathcal{A}_h)) \cup (\mathcal{F} \cap \mathcal{A}_h)= \emptyset \cup \mathcal{A}_h$. It follows from $\mathcal{A}_h \subset \mathcal{F} \cap \Upsilon_h$ and $\mathcal{F} \cap \Upsilon_h \subset \mathcal{A}_h$ that $\mathcal{F} \cap \Upsilon_h = \mathcal{A}_h$. Hence, $u_{\mathcal{F}}^{-1}(0)=\mathcal{A}_h$. Consequently, every maximal solution of the hybrid system \eqref{hybrid_sys} converges to the largest weakly invariant subset $U_R^{-1}(0)\cap \mathcal{A}_h = \mathcal{A}_h$. Since every maximal solution of the hybrid closed-loop system (\ref{hybrid_sys}) is bounded, $G(x_h) \in \mathcal{F} \cup \mathcal{J}$ for every $x_h \in \mathcal{J}$, and $F(x_h)\subset T_\mathcal{F}(x_h)$, for every $x_h \in \mathcal{F}\setminus \mathcal{J}$, where $T_\mathcal{F}(x_h)$ denotes the tangent cone to $\mathcal{F}$ at the point $x_h$, according to \cite[Proposition 6.10]{goebel2012hybrid}, one can conclude that every maximal solution of the hybrid closed-loop system (\ref{hybrid_sys}) is complete. This, together with Lemma \ref{lem_hbc}, allows us to conclude that the set $\mathcal{A}_h$ is globally asymptotically stable for the hybrid closed-loop system (\ref{hybrid_sys}). This completes the proof.
\end{proof}
\begin{rmk}
    Note that the design of our proposed hybrid distributed attitude observer \eqref{TH_observer_dynamics_1}-\eqref{TH_observer_dynamics_2} is based on a generic potential function $U_R$, defined on $\mathcal{S}_h$, with respect to $\mathcal{A}_h$. It is also important to note that the flow set $\mathcal{F}$ and the jump set $\mathcal{J}$, given in \eqref{network_f_j_set}, depend on the parameters $\delta$ and $\Xi$. These parameters, along with the potential function $U_R$, must be carefully designed to satisfy Condition \ref{hybrid_ass}.
\end{rmk}

\indent In the next section, we will provide the explicit structure of our proposed hybrid estimation scheme by first presenting the potential function and specifying the set of parameters $\mathcal{P}$ in which Condition \ref{hybrid_ass} is satisfied. We then provide the explicit form of our proposed hybrid distributed attitude observer in terms of relative attitude measurements and attitude estimates.\\

\subsubsection{Explicit Hybrid Distributed Attitude Observer Design Using Relative Attitude Measurements}
Let us begin this section by introducing the potential function $U_R$ and some useful related properties. Consider the potential function $U_R(x_h)$, given in \eqref{potential_fct}, where $U(\bar R_k, \xi_k)$ is defined as follows:
\begin{equation}
    U(\bar R_k,\xi_k):=\text{tr}\Big(A\left(I_3-\bar R_k\mathcal{R}_\alpha(\xi_k,u)\right)\Big)+\frac{\gamma}{2}\xi_k^2,
\end{equation}
with $A \in \mathbb{R}^{3\times 3}$ satisfies Assumption \ref{A_ass}, $u\in \mathbb{S}^2$ is a constant unit vector and $\gamma$ is a positive scalar. Note that $U_R$ is an extension of the potential function $U(R_e,\theta)$ proposed in \cite{miaomiao_TAC2022} for a single agent attitude control design. In the following proposition, we will derive the gradient of $U_R$ with respect to $\bar R_k$ and $\xi_k$ and give the set of all its critical points.
\begin{pro}\label{chr_pf}
    Consider the potential function $U_R(x_h)$, where Assumption \ref{A_ass} is satisfied. Then, the following statements hold:
    \begin{itemize}
        \item $ \psi \left(\bar R_k^T \nabla_{\bar R_k}U_R\right)= \mathcal{R}_\alpha(\xi_k,u) \psi(A\bar R_k\mathcal{R}_\alpha(\xi_k,u))$ for all $k \in \mathcal{M}$.
        \item $\nabla_{\xi_k}U_R=\gamma \xi_k +2 u^T \psi(A\bar R_k\mathcal{R}_\alpha(\xi_k,u))$ for all $k \in \mathcal{M}$.
        \item $\Upsilon_h= \{x_h \in \mathcal{S}_h: \forall k\in \mathcal{M}, ~\mathbb{P}_a(A \bar R_k)=0, \xi_k=0\}$.
        \item $\mathcal{A}_h \subset \Upsilon_h$.
        %\item $\Upsilon_h= \mathcal{A}_h \cup \{x_h \in \mathcal{S}_h: \forall k\in \mathcal{M}, ~\bar R_k=\mathcal{R}_\alpha(\pi,v), \xi_k=0, ~ v \in \mathcal{E}(A)\}$
    \end{itemize}
\end{pro}
\begin{proof}
    %The proof can be found in \cite[Proposition 1]{Mouaad_arxiv_TAC2023}.\\
    The time derivative of $U_R$, along the trajectories of the hybrid closed-loop system (\ref{hybrid_sys}), is given by
    
    \begin{align}
    \dot{U}_R(x_h)=&-\sum_{k=1}^{M}\text{tr}\big(A\bar{R}_k\mathcal{R}_\alpha(\xi_k,u)[k_R\mathcal{R}_\alpha(\xi_k,u)^T\bar{\sigma}_k\nonumber\\&+\dot{\xi}_k u]^\times\big)+\gamma \sum_{k=1}^{M}\dot{\xi}_k \xi_k.
\end{align}
Using the following identities: $\text{tr}\left(B[x]^\times\right)=\text{tr}\left(\mathbb{P}_a(B)[x]^\times\right)$ and $\text{tr}\left([x]^\times [y]^\times\right)=-2x^Ty$, $\forall x, y\in \mathbb{R}^3$ and $\forall B \in \mathbb{R}^{3\times3}$, one obtains
\begin{align}\label{U_der_1}
    \dot{U}_R(x_h) =& 2k_R\sum_{k=1}^{M}\bar{\sigma}_k^T \mathcal{R}_\alpha(\xi_k,u)\psi\big(A\bar{R}_k\mathcal{R}_\alpha(\xi_k,u)\big)\nonumber\\
    &+\sum_{k=1}^{M}\dot{\xi}_k\Big(\gamma \xi_k +2 u^T \psi\big(A\bar{R}_k\mathcal{R}_\alpha(\xi_k,u)\big)\Big).
\end{align}
It follows from (\ref{U_grd}) and (\ref{U_der_1}) that $ \psi \left(\bar R_k^T \nabla_{\bar R_k}U_R\right)= \mathcal{R}_\alpha(\xi_k,u) \psi(A \bar R_k\mathcal{R}_\alpha(\xi_k,u))$ and $\nabla_{\xi_k}U_R=\gamma \xi_k +2 u^T \psi(A\bar R_k\mathcal{R}_\alpha(\xi_k,u))$ for all $k \in \mathcal{M}$. Considering the last two expressions ($\psi \left(\bar R_k^T \nabla_{\bar R_k}U_R\right)$ and $\nabla_{\xi_k}U_R$) and the definition of the set of all critical points of $U_R$, given in \eqref{equi_set_}, one can conclude that $\Upsilon_h$ is the set of all critical points of $U_R$. Moreover, one can conclude that $\mathcal{A}_h \subset \Upsilon_h$. This completes the proof.
\end{proof}
\begin{rmk}\label{rmk_explicit}
    Note that the expressions for the gradients of $U_R$ with respect to $\bar R_k$ and $\xi_k$, given in Proposition \ref{chr_pf}, are derived in terms of relative attitude errors $\bar R_k$. These relative attitude errors can be constructed from the relative orientation measurements \eqref{measurement_model_R} and the estimated orientations as follows:  $\bar R_k = \hat R_j R_{ij}^T \hat R_i^T$, for every $(i,j) \in \mathcal E$ such that $\{k\}=\mathcal{M}_i^+\cap\mathcal{M}^-_j$. 
\end{rmk}

Now, consider the set of parameters $\mathcal{P}:=\{\Xi, A, u, \gamma, \delta\}$. The next proposition gives the possible choices of parameters in the set $\mathcal{P}$ in which Condition \ref{hybrid_ass} is satisfied. 
\begin{pro}\label{pro_set}
Consider the potential function $U_R$. Then, Condition \ref{hybrid_ass} holds under the following set of parameters $\mathcal{P}$:
\begin{equation} {\mathcal{P}}: \begin{cases} {\Xi = \left\{ {\left| {{\xi _i}} \right| \in (0,\pi ],i = 1, \cdots ,m} \right\}} \\ {A:0 < \lambda_1 \leq \lambda_2 < \lambda _3} \\ {u = {\alpha _1}q_1 + {\alpha _2}q_2 + {\alpha _3}q_2} \\ {\gamma < \frac{{4{\Delta ^{\ast}}}}{{{\pi ^2}}}} \\ {0<\delta < \left( {\frac{{4{\Delta ^{\ast}}}}{{{\pi ^2}}} - \gamma } \right)\frac{{\xi_M^2}}{2},{\xi_M}: = \mathop {\max }\limits_{\xi \in \Xi } \left| \xi \right|} \end{cases}\end{equation}
where $\alpha_1^2+\alpha_2^2+\alpha_3^2=1$ and $\Delta^*>0$ are given as follows:
\begin{itemize}
    \item If $\lambda_1=\lambda_2$, $\alpha_3^2=1-\frac{\lambda_2}{\lambda_3}$ and $\Delta^*=\lambda_1(1-\frac{\lambda_2}{\lambda_3})$.
    \item If $\lambda_2\geq \frac{\lambda_1 \lambda_3}{\lambda_3-\lambda_1}$, $\alpha_i^2=\frac{\lambda^A_i}{\lambda_2+\lambda_3}$,$i\in \{2,3\}$ and $\Delta^*=\lambda_1$.
    \item If $\lambda_1<\lambda_2<\frac{\lambda_1\lambda_3}{\lambda_3-\lambda_1}$, $\alpha_i^2=1-\frac{4\prod_{l\neq i}\lambda_l}{\sum_{l=1}^{3}\sum_{k\neq l}^{3}\lambda_l\lambda_k}, \forall i\in \{1, 2, 3\}$, and $\Delta^*=\frac{4\prod_{l}\lambda_l}{\sum_{l=1}^{3}\sum_{k\neq l}^{3}\lambda_l\lambda_k}$.
\end{itemize}
where $(\lambda_i, q_i)$ denoting the $i$-th pair of eigenvalue-eigenvector of matrix A.
\end{pro}
\begin{proof}
Following the same arguments given in the proof of \cite[Proposition 2]{miaomiao_TAC2022}, one can prove Proposition \ref{pro_set}.
\end{proof}

Let us conclude this section by giving the explicit form of our proposed hybrid distributed attitude estimation scheme. From Proposition \ref{chr_pf}, Proposition \ref{pro_set}, and the fact in Remark \ref{rmk_explicit}, one can explicitly derive the proposed hybrid distributed attitude observer, given in \eqref{TH_observer_dynamics_1}-\eqref{TH_observer_dynamics_2}, in terms of relative attitude measurements and attitude estimates as follows: 
\begin{equation}
    \underbrace{
                  \begin{aligned}\label{R_obs_f}
                    \dot{\hat{R}}_i&=\hat{R}_i[\omega_i-k_R\hat{R}_i^T \sigma_i]^\times \\
                    \dot{\xi}_k&=-k_\xi\left(\gamma \xi_k+2u^T\psi\big(A\bar{R}_k\mathcal{R}_a(\xi_k,u)\big)\right)\\
                    \sigma_i&=-\Biggl(\sum_{j\in\mathcal{I}_i}\mathcal{R}_a(\xi_p,u)\psi(A\hat R_j R_{ij}^T \hat R_i^T \mathcal{R}_a(\xi_p,u))\\
                    &~~~~~~~~~~~~~~~~~~~~~~~~~+\sum_{j\in\mathcal{O}_i}\psi(A\mathcal{R}_a(\xi_l,u)^T\hat R_j R_{ij}^T \hat R_i^T )\Biggl)
                 \end{aligned}
                }_{x_h \in \mathcal{F}_i}
\end{equation}
\begin{equation}
                 \underbrace{
                \begin{aligned}\label{R_obs_j}
                    \hat R_i^+ &= \hat R_i\\
                    \xi_k^+&\in 
                    \begin{cases}
                    \xi_k \hspace{1.3cm}\text{if} \hspace{0.5cm} U(\bar{R}_k, \xi_k)-U(\bar{R}_k, \xi^*_k) \leq \delta\\
                    \xi^*_k \hspace{1.3cm}\text{if} \hspace{0.5cm} U(\bar{R}_k, \xi_k)-U(\bar{R}_k, \xi^*_k) \geq \delta
                    \end{cases}
                \end{aligned}
                }_{x_h \in \mathcal{J}_i}
\end{equation}
where $i\in\mathcal{V}$, $ k \in \mathcal{M}_i^+$, $\{p\}=\mathcal{M}_i^+\cap \mathcal{M}_j^- \in \mathcal{M}$ and $\{l\}=\mathcal{M}_i^-\cap \mathcal{M}_j^+ \in \mathcal{M}$. Note that $\mathcal{N}_i=\mathcal{I}_i\cup\mathcal{O}_i$, with $\mathcal{I}_i:=\{j\in \mathcal{N}_i: j \hspace{0.1cm} \text{is the tail of the edge} \hspace{0.1cm} (i,j) \in \mathcal{E}\}$ and $\mathcal{O}_i:=\{j\in \mathcal{N}_i: j \hspace{0.1cm} \text{is the head of the edge} \hspace{0.1cm} (i,j) \in \mathcal{E}\}$.
\begin{rmk}
For the implementation of our proposed hybrid distributed attitude observer \eqref{R_obs_f}-\eqref{R_obs_j}, we assume that the dynamics of $\xi_k$ are implemented at agent $i$, and agent $j$ receives the information about $\xi_k$ from agent $i$, according to Assumption \ref{measurement_ass}, for every $(i,j) \in \mathcal E$ such that $\{k\}=\mathcal{M}_i^+\cap\mathcal{M}^-_j$. 
\end{rmk}

\section{Application to distributed pose estimation}\label{s5}
As an application, we will use the proposed hybrid distributed attitude observer \eqref{R_obs_f}-\eqref{R_obs_j} to globally estimate the poses of $N$ agents, evolving in a $3$-dimensional space, up to a common constant translation and orientation, based on individual body-frame linear velocity measurements and local relative time-varying bearing measurements. The idea is to map the local relative (time-varying) bearing measurements, via the estimated attitudes obtained from \eqref{R_obs_f}-\eqref{R_obs_j}, into the inertial frame so they can be used for the hybrid distributed position estimation scheme.

Consider a network of $N$ agents. Let the following kinematic equations describe the agents' motion:
\begin{eqnarray}
    \dot{R}_i &=& R_i[\omega_i]^{\times}\label{pb1_rotation}\\
    \dot{p}_i &=& v_i, \label{p_dynamics}
\end{eqnarray}
where $p_i \in \mathbb{R}^3$ and $v_i \in \mathbb{R}^3$ are the position and velocity of agent $i$, respectively, expressed in the inertial frame and $i \in \mathcal{V}$. Let $p:= \left[p_1^T, p_2^T, \hdots, p_N^T\right]^T$. The graph $\mathcal{G}$ together with the stack position vector $p$ define the formation $\mathcal{G}(p(t))$.

The measurement of the local relative bearing between agent $i$ and agent $j$ is given by
\begin{equation}\label{b_measurement_model}
    b_{ij}^i(t):=R_i^T b_{ij}(t),
\end{equation}
where $b_{ij}(t):=\frac{p_j(t)-p_i(t)}{||p_j(t)-p_i(t)||}$ and $b_{ij}^i(t)$ are the relative bearing measurements between agent $i$ and agent $j$ expressed in the inertial frame and the body-attached frame of agent $i$, respectively. Recall that the agents orientations are time-varying, but for the sake of simplicity, the time argument $t$ is omitted in the expression of $R_i$ in \eqref{b_measurement_model}. Before presenting our proposed distributed position estimation law, we will introduce some important definitions and assumptions.
\begin{Definition}\cite{Tang_CDC2020}
    Consider the formation $\mathcal{G}(p(t))$ with an arbitrary orientation of the graph $\mathcal{G}$. Define the matrix $\textit{\textbf{L}}$ and bearing Laplacian matrix $\textit{\textbf{L}}_B$ as $\textit{\textbf{L}} := L\otimes I_3$ and $\textit{\textbf{L}}_B(t) :=\textit{\textbf{H}} \text{diag}(P_{b_k(t)})\textit{\textbf{H}}^T$, respectively, where $b_k$ is the bearing vector corresponding to the edge $k$, and $\textit{\textbf{H}}=H \otimes I_3$ \big($H$ and $L$ are the incidence and the Laplacian matrices, respectively, corresponding to the graph $\mathcal{G}$\big). The bearing Laplacian matrix is called persistently exciting (PE) if there exists $T>0$ and $\mu >0$ such that:
    \begin{align}
     \int_t^{t+T} \textit{\textbf{L}}_B(\tau) d\tau \geq \mu \textit{\textbf{L}},
   \end{align}
    for all $t$.
\end{Definition}

\begin{Definition}\cite{Tang_CDC2020}
    The formation $\mathcal{G}(p(t))$ is bearing persistently exciting if the graph $\mathcal{G}$ is connected and its bearing Laplacian matrix is PE.
\end{Definition}
Now, consider the following assumptions:

\begin{assumption}\label{measurement_ass_2}
  Each agent $i$ in the formation measures the local relative bearings $b_{ij}^i(t)$ with respect to its neighboring agents $j\in \mathcal{N}_i$.
\end{assumption} 

\begin{assumption}\label{pe_ass}
  The formation $\mathcal{G}(p(t))$ is bearing persistently exciting.
\end{assumption}

\begin{assumption}\label{b_ass}
   As the formation evolves over time, there is no collision between agents.
\end{assumption}

\begin{assumption}\label{measurement_avialble}
  The body-frame linear velocity of each agent is bounded and available.
\end{assumption}

With all these ingredients, we propose the following distributed hybrid pose estimation law:
\begin{align}
             &\underbrace{
                  \begin{aligned}\label{p_obs_f}
                    \dot{\hat{R}}_i&=\hat{R}_i[\omega_i-k_R\hat{R}_i^T \sigma_i]^\times \\
                    \dot{\hat{p}}_i &= \hat{R}_i v_i^i-k_p \sum_{j\in \mathcal{N}_i} \hat{R}_i \left(P_{b^i_{ij}(t)} \hat{R}_i^T\hat{p}_i-R_{ij}P_{b^j_{ji}(t)} \hat{R}_j^T\hat{p}_j\right)\\
                    &~~~~~~~~~~~~~~~~~~~~~~~~~~~~~-k_R[\sigma_i]^\times\hat{p}_i\\
                    \dot{\xi}_k&=-k_\xi\left(\gamma \xi_k+2u^T\psi\big(A\bar{R}_k\mathcal{R}_a(\xi_k,u)\big)\right)
                 \end{aligned}
                }_{x_h \in \mathcal{F}_i}
                \\
             &\underbrace{
                \begin{aligned}\label{p_obs_j}
                    \hat R_i^+ &= \hat R_i\\
                    \hat{p}_i^+ &= \hat{p}_i\\
                    \xi_k^+&\in 
                    \begin{cases}
                    \xi_k \hspace{1.3cm}\text{if} \hspace{0.5cm} U(\bar{R}_k, \xi_k)-U(\bar{R}_k, \xi^*_k) \leq \delta\\
                    \xi^*_k \hspace{1.3cm}\text{if} \hspace{0.5cm} U(\bar{R}_k, \xi_k)-U(\bar{R}_k, \xi^*_k) \geq \delta
                    \end{cases}
                \end{aligned}
                }_{x_h \in \mathcal{J}_i}
\end{align}
where $i \in \mathcal{V}$, $k_p>0$, $\hat{p}_i\in \mathbb{R}^3$ is the estimate of $p_i$ and $v_i^i$ is the body-frame linear velocity of agent $i$. Note that the expression of $\sigma_i$ is given in \eqref{R_obs_f}. Now, let us state the main result of this section. 
\begin{thm}
    Suppose Assumptions \ref{hybrid_ass}-\ref{measurement_avialble} and Condition \ref{hybrid_ass} hold. Then, with $k_R, k_\xi, k_p >0$, the distributed hybrid observer \eqref{p_obs_f}-\eqref{p_obs_j} estimates the agents' poses governed by the kinematics \eqref{pb1_rotation}-\eqref{p_dynamics} globally asymptotically up to a common constant translation and orientation. 
\end{thm}

\begin{proof}
     Since the local relative bearings with respect to neighboring agents and the body-frame linear velocity of each agent are available, according to Assumptions \ref{measurement_ass_2}, \ref{measurement_avialble}, and using the facts that $v_i^i = R_i^T v_i$ and $P_{b^i_{ij}(t)}=R_i^T P_{b_{ij}(t)}R_i$, it follows from \eqref{p_obs_f} that the position estimate $\hat p_i$ of each agent $i \in \mathcal{V}$, during the flows, can be rewritten as follows:
\begin{align}\label{Translation_observer_dynamics}
    \dot{\hat{p}}_i &= \tilde{R}^T_i v_i-k_p \tilde{R}^T_i \sum_{j\in \mathcal{N}_i} P_{b_{ij}(t)}  \left(\tilde{R}_i\hat{p}_i-\tilde{R}_j\hat{p}_j\right)-k_R[\sigma_i]^\times\hat{p}_i.
\end{align}
Note that, under Assumption \ref{b_ass}, one has $||p_i-p_j||\neq0$. Consequently, the bearing  $b_{ij}(t)$, for every $(i,j) \in \mathcal E$, is well defined for all $t\geq0$. Define the position estimation error as $\tilde{p}_i:=\tilde{R}_i\hat{p}_i - p_i$, for every $i\in\mathcal{V}$. In view of (\ref{p_dynamics}), (\ref{absolute_attitude_error}) and the above equation, the hybrid dynamics of $\tilde{p}_i$ are given by 
\begin{align}
             &\underbrace{
                  \begin{aligned}\label{p_obs_err_f}
                    \dot{\tilde{p}}_i =k_p\sum_{j\in \mathcal{N}_i} P_{b_{ij}(t)}(\tilde{p}_j-\tilde{p}_i)
                 \end{aligned}
                }_{x_h \in \mathcal{F}_i}
                \\
             &\underbrace{
                \tilde{p}_i^+ = \tilde{p}_i
                }_{x_h \in \mathcal{J}_i}\label{p_obs_err_j}
\end{align}
for every $i\in\mathcal{V}$. The equations \eqref{p_obs_err_f} and \eqref{p_obs_err_j} were obtained using the fact that $P_{b_{ij}}(p_j-p_i)=0$ and $\tilde p^+=\tilde R_i^+ \hat p_i^+-p_i=\tilde R_i \hat p_i-p_i=\tilde p_i$, respectively. Thanks to the definition of the position estimation error $\tilde p_i$, for every $i \in \mathcal{V}$, which allows deriving the above hybrid distributed estimation error dynamics that are independent of the attitude estimates provided by the hybrid attitude observer \eqref{R_obs_f}-\eqref{R_obs_j}. Now, letting $e_i := \tilde p_i-\frac{1}{N} \sum_{n=1}^N \tilde p_n(0)$, it follows from \eqref{p_obs_err_f}-\eqref{p_obs_err_j} that 
\begin{align}
             &\underbrace{
                  \begin{aligned}\label{e_obs_err_f}
                    \dot{e}_i =k_p\sum_{j\in \mathcal{N}_i} P_{b_{ij}(t)}(e_j-e_i)
                 \end{aligned}
                }_{x_h \in \mathcal{F}_i}
                \\
             &\underbrace{
                e_i^+ = e_i
                }_{x_h \in \mathcal{J}_i}\label{e_obs_err_j}
\end{align}
Define the new state space $\bar{\mathcal{S}}_h:= SO(3)^M \times \mathbb{R}^M \times \mathbb{R}^{3N} \times \mathbb{R}$ and the new state $\bar x_h := \left(\bar R_1, \hdots, \bar R_M, \xi_1, \hdots, \xi_M, e_1,\hdots, e_N, t\right)\in \bar{\mathcal{S}}_h$. From \eqref{hybrid_sys}, \eqref{R_obs_f}-\eqref{R_obs_j} and \eqref{e_obs_err_f}-\eqref{e_obs_err_j}, one obtains the following extended hybrid multi-agent closed-loop system:
\begin{equation}\label{T_hybrid_sys}
\bar{\mathcal{H}}:\begin{cases} {\dot{\bar{x}}_h = \bar F(\bar x_h)}&{\bar x_h \in {\bar{\mathcal{F}}}} \\ {{\bar x_h^ + } \in \bar G(\bar x_h)}&{\bar x_h \in {\bar{\mathcal{J}}}} \end{cases}
\end{equation}
where $\bar{\mathcal{F}}:=\{\bar x_h \in \bar{\mathcal{S}}_h:~x_h \in \mathcal{F}\}$, $\bar{\mathcal{J}}:=\{\bar x_h \in \bar{\mathcal{S}}_h:~x_h \in \mathcal{J}\}$, and the flow and jump maps are given by 
\begin{equation}
    \bar F(\bar x_h) := \left[ {\begin{array}{c} {k_R \bar R_1[\bar \sigma_1]^{\times}} \\ \vdots\\ {k_R \bar R_M[\bar \sigma_M]^{\times}}\\{-k_\xi\left(\gamma \xi_1+2u^T\psi\big(A\bar{R}_1\mathcal{R}_a(\xi_1,u)\big)\right)} \\ \vdots\\ {-k_\xi\left(\gamma \xi_M+2u^T\psi\big(A\bar{R}_M\mathcal{R}_a(\xi_M,u)\big)\right)}\\{k_p\sum_{j\in \mathcal{N}_1} P_{b_{1j}(t)}(e_j-e_1)} \\ \vdots\\ {k_p\sum_{j\in \mathcal{N}_N} P_{b_{Nj}(t)}(e_j-e_N)}\\{1} \end{array}}\right] \nonumber
\end{equation}
and {\small $$\bar G(\bar x_h) :=\left(\bar R_1, \hdots, \bar R_M, \{\xi_1, \xi_1^*\}, \hdots, \{\xi_M, \xi_M^*\}, e_1, \hdots, e_N, t\right).$$}Recall that $\bar \sigma = \textit{\textbf{H}}^T \sigma$ where $\sigma_i$, for every $i \in \mathcal{V}$, is given in \eqref{R_obs_f}, $\xi^*_k := \text{arg} \underset{\bar{\xi}_k \in \Xi}{\text{min}} U(\bar{R}_k, \bar{\xi}_k)$, and the sets $\mathcal{F}$ and $\mathcal{J}$ are defined in \eqref{network_f_j_set}. Note that, since the relative bearings are time-varying, we consider the time $t$ as an additional state variable to make the overall system \eqref{T_hybrid_sys} autonomous. Note also that $\bar{\mathcal{F}}\cup \bar{\mathcal{J}} = \bar{\mathcal{S}}_h$ and the hybrid closed-loop system \eqref{T_hybrid_sys} satisfies the hybrid basic conditions. Next, we will prove that the set of desired equilibria $\bar{\mathcal{A}}_h:=\{\bar x_h \in \bar{\mathcal{S}}_h:~\bar R_1=I_3, \hdots, \bar R_M=I_3, \xi_1=0, \hdots, \xi_M=0, e_1=0, \hdots, e_N = 0\}$ of the hybrid closed-loop system \eqref{T_hybrid_sys} is globally asymptotically stable and the number of jumps is finite. Define $e:=\left[e_1^T, e_2^T, \hdots, e_N^T\right]^T \in \mathbb{R}^{3 N}$. According to the flows of \eqref{T_hybrid_sys}, one has
\begin{align}
    \dot{e} = -k_p \textit{\textbf{L}}_B(t) e.
\end{align}
\noindent Notice that $(\textbf{1}_N\otimes I_3)^T e(0) = 0$ and $(\textbf{1}_N\otimes I_3)^T \dot{e}=0$.  Moreover, since the bearings $b_{ij}(t)$, for every $(i, j)\in \mathcal{E}$, are bounded, $\forall t\geq0$, it follows that $\textit{\textbf{L}}_B(t)$ is bounded. With all of these and Assumption \ref{pe_ass}, it follows from \cite[Lemma 5]{antonio_2002} that
\begin{equation}\label{ineq}
    ||e(t)||^2 \leq ||e(0)||^2 \text{e}^{-\beta t},
\end{equation}
where $\beta$ is a positive scalar. Again, using the fact that $\textit{\textbf{L}}_B(t)$ is bounded, $\forall t\geq 0$, together with inequality \eqref{ineq}, it follows from the converse theorem \cite{khalil2002nonlinear} that there is exist a real-valued function $U_e: [0, \infty) \times \mathbb{R}^{3N} \rightarrow \mathbb{R}$ such that the following inequalities hold:
\begin{align}
    c_1 ||e||^2&\leq U_e (t, e) \leq c_2 ||e||^2 \label{posi_defi}\\
    \dot{U}_e (t, e) &\leq - c_3 ||e||^2,\label{dot_posi_defi}
\end{align}
where $c_1, c_2$ and $c_3$ are positive constants. Now, consider the following Lyapunov function candidate:
\begin{align}
    \mathcal{L}(\bar x_h)&= U_R(x_h) + U_e(t, e) =\sum_{k=1}^{M}U(\bar R_k,\xi_k) + U_e(t, e)\nonumber\\
    &=\hspace{-0.1cm}\sum_{k=1}^{M}\hspace{-0.05cm} \left(\text{tr}\Big(A\left(I_3-\bar R_k\mathcal{R}_\alpha(\xi_k,u)\right)\Big)+\frac{\gamma}{2}\xi_k^2\right)+U_e (t, e).\nonumber
\end{align}
Recall that $U_R(x_h)$ is a potential function on $\mathcal{S}_h$ with respect to $\mathcal{A}_h$. Combining this fact with inequality \eqref{posi_defi}, one can verify that $\mathcal{L}(\bar x_h)$ is positive definite on $\bar{\mathcal{S}}_h$ with respect to $\bar{\mathcal{A}}_h$. The time-derivative of $\mathcal{L}(\bar x_h)$, along the trajectories generated by the flows of the hybrid closed-loop system \eqref{T_hybrid_sys}, is given by
\begin{align}
    \dot{\mathcal{L}}(\bar x_h)=&-2k_R||\textit{\textbf{H}} \Psi_\nabla^{\bar R}||^2-k_\xi || \Psi_\nabla^\xi||^2+\dot U_e(t, e)\nonumber\\
    \leq&-2k_R||\textit{\textbf{H}} \Psi_\nabla^{\bar R}||^2-k_\xi || \Psi_\nabla^\xi||^2-c_3||e||^2\label{inq_mm} \\
    \leq& 0\label{flow_T_hyb_2},
\end{align}
where the elements of the vectors $\Psi_\nabla^{\bar R}$ and $\Psi_\nabla^\xi$ are explicitly given in Proposition \ref{chr_pf}. Inequality \eqref{inq_mm} was obtained using the fact given in  \eqref{dot_posi_defi}. This implies that $\mathcal{L}(\bar x_h)$ is non-increasing over the flows of \eqref{T_hybrid_sys}. Furthermore, one has 
\begin{align}
    \mathcal{L}(\bar x_h)-\mathcal{L}(\bar x_h^+)&=U_R(x_h)-U_R(x_h^+)\label{jump_T_hyb_1}\\
    =&\sum_{k=1}^{M}\left(U(\bar{R}_k, \xi_k)-U(\bar{R}_k^+, \xi_k^+)\right)\nonumber\\
    \geq&\delta\label{jump_T_hyb_2},
\end{align}
where we have used the fact that $U_e(t, e)-U_e(t^+, e^+)=0$ to obtain \eqref{jump_T_hyb_1}. Inequality \eqref{jump_T_hyb_2} shows the strict decrease of $\mathcal{L}(\bar x_h)$ over the jumps of \eqref{T_hybrid_sys}. Moreover, using arguments similar to the ones used in the proof of Theorem \ref{theorem2}, one can show that the set $\bar{\mathcal{A}}_h$ is stable, every maximal solution of the hybrid closed-loop system \eqref{T_hybrid_sys} is bounded, and the number of jumps is finite.\\
Now, let us prove the global asymptotic stability of the set $\bar{\mathcal{A}}_h$. Following the same steps as in the proof of Theorem \ref{theorem2}, with{\small
\begin{align}
     u_{\bar{\mathcal{F}}}(\bar x_h) &:=\begin{cases}
                     -2k_R||\textit{\textbf{H}} \Psi_\nabla^{\bar R}||^2-k_\xi || \Psi_\nabla^\xi||^2-c_3||e||^2~~\text{if}~\bar x_h \in \bar{\mathcal{F}},\\
                     -\infty~~~~~~~~~~~~ \text{otherwise},
                \end{cases}\\
     u_{\bar{\mathcal{J}}}(\bar x_h) &:=\begin{cases}
                    &-\delta~~~~~~~~~~~~ \text{if}~\bar x_h \in \bar{\mathcal{J}},\\
                    & -\infty~~~~~~~~~~ \text{otherwise},
                \end{cases}
\end{align}}one can show that every maximal solution of the hybrid system \eqref{T_hybrid_sys} converges to the largest weakly invariant subset $\bar{\mathcal{A}}_h$. Furthermore, using the fact that every maximal solution of \eqref{T_hybrid_sys} is bounded, $\bar G(\bar x_h) \in \bar{\mathcal{F}} \cup \bar{\mathcal{J}}$ for every $\bar x_h \in \bar{\mathcal{J}}$, and $\bar F(\bar x_h)\subset T_{\bar{\mathcal{F}}}(\bar x_h)$, for every $\bar x_h \in \bar{\mathcal{F}}\setminus \bar{\mathcal{J}}$, according to \cite[Proposition 6.10]{goebel2012hybrid}, one can verify that every maximal solution of \eqref{T_hybrid_sys} is complete. This, together with the fact that \eqref{T_hybrid_sys} satisfies the basic hybrid conditions, implies that the set $\bar{\mathcal{A}}_h$ is globally asymptotically stable for the hybrid system \eqref{T_hybrid_sys}. This completes the proof.
\end{proof}
\begin{rmk}
Note that the global asymptotic stability of the set $\bar{\mathcal{A}}_h$ for the hybrid system \eqref{T_hybrid_sys} implies that the poses of the $N$ agents can be estimated up to a constant (common) translation and orientation, which can be determined if at least one agent in the formation has access to its absolute attitude and position (\ie, having a leader in the group), in which case the agents' poses can be estimated without ambiguity.
\end{rmk}
\begin{rmk}
Unlike most existing works (see, \eg, \cite{li_2020,lee_2019,Zhao_auto2015}), our proposed hybrid distributed estimation scheme can globally estimate the individual poses of multi-agent rigid body systems subjected to time-varying translational and rotational motion. Furthermore, in contrast to \cite{Tang_CDC2020,Tang_IFAC2020}, our scheme relies on local time-varying bearing measurements. Thanks to the hybrid distributed attitude observer \eqref{R_obs_f}-\eqref{R_obs_j} which was instrumental in designing this scheme with a global asymptotic stability guarantees.   
\end{rmk}

\section{SIMULATION}\label{s6}
In this section, we will present some numerical simulations to illustrate the performance of the continuous attitude observer \eqref{observer_dynamics}, \eqref{C_correcting_term}, the hybrid attitude observer \eqref{R_obs_f}-\eqref{R_obs_j}, as well as the hybrid distributed pose estimation law \eqref{p_obs_f}-\eqref{p_obs_j}.\\
We consider a five-agent system in a three-dimensional space that forms a square pyramid rotating around the $z$-axis (see Figure \ref{Real_formation}) with the following positions: $p_i(t)=R^T(t)p_i(0)$ where $R(t)=[\cos{\frac{\pi}{6}t}~-\sin{\frac{\pi}{6}t}~0;~\sin{\frac{\pi}{6}t}~\cos{\frac{\pi}{6}t}~0;~0~0~1]$, $p_1(0)=[-2~-2~-2]^T$, $p_2(0)=[2~-2~-2]^T$, $p_3(0)=[-2~2~-2]^T$, $p_4(0)=[2~2~-2]^T$ and $p_5(0)=[0~0~0]^T$. The rotational motions of the agents are driven by the following angular velocities: $\omega_1=[1~-2~1]^T$, $\omega_2(t)=[-\cos{3t}~1~\sin{2t}]^T$, $\omega_3(t)=[-\cos{t}~1~\sin{2t}]^T$, $\omega_4(t)=[-\cos{2t}~1~\sin{5t}]^T$ and $\omega_5=[1.5~4~5]^T$, where the initial rotations of all agents are chosen to be the identity, \ie, $R_i(0)=I_3$, for every $i \in \mathcal{V}$. 

\begin{figure}[H]
    \centering
    \includegraphics[width=0.47\linewidth]{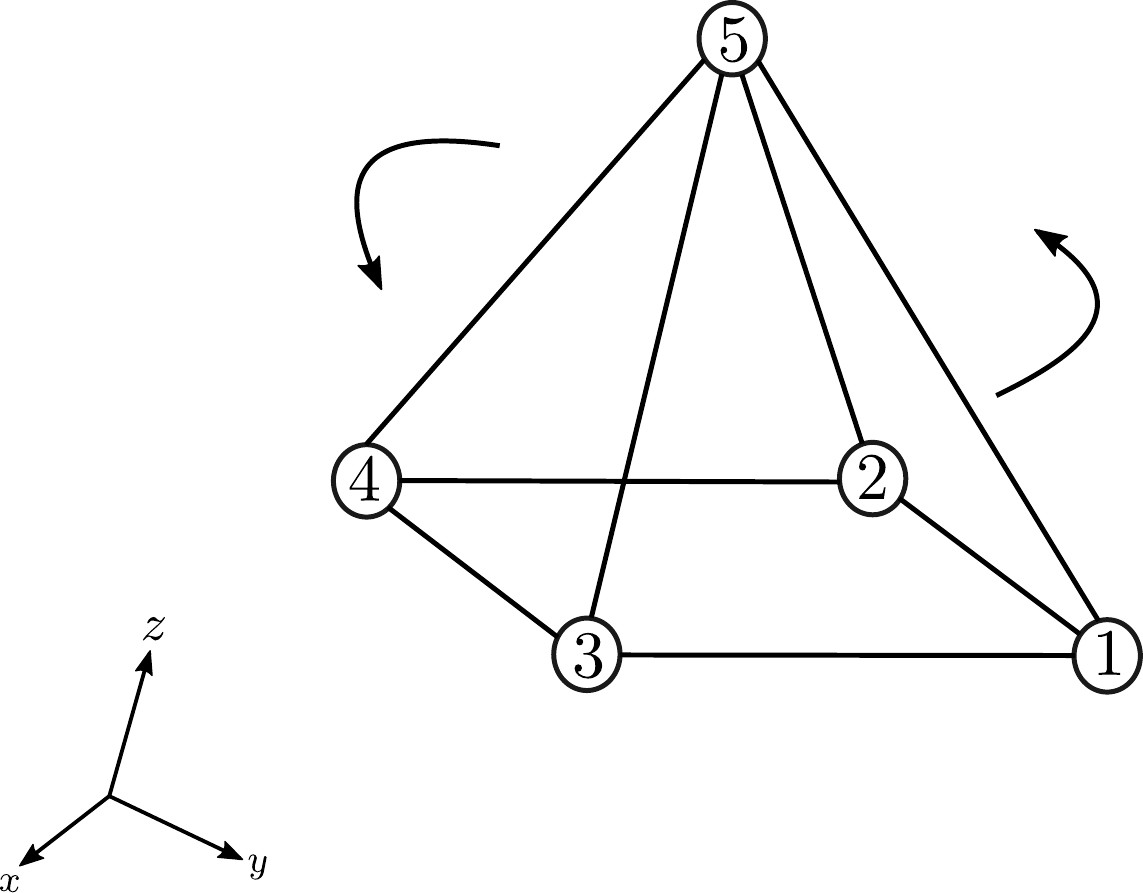}
    \caption{The actual formation, $\mathcal{G}(p)$, in $\mathbb{R}^3$.}
    \label{Real_formation}
\end{figure}

\indent We use an undirected graph topology to describe the interactions between these agents (see Figure \ref{undirected_graph}). Accordingly, the neighbors sets are given as $\mathcal{N}_1=\{2\}$, $\mathcal{N}_2 = \{1, 3\}$, $\mathcal{N}_3 = \{2, 4\}$, $\mathcal{N}_4 = \{3, 5\}$ and $\mathcal{N}_5 = \{4\}$.

\begin{figure}[h]
    \centering
    \subfigure[]{\includegraphics[width=0.47\linewidth]{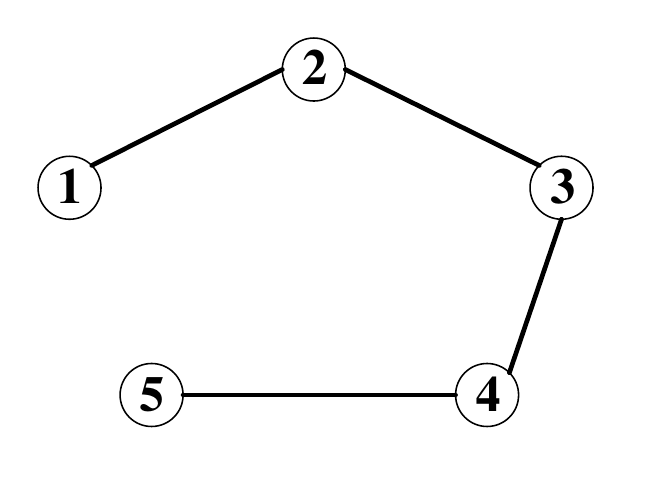}\label{undirected_graph}} 
    \subfigure[]{\includegraphics[width=0.47\linewidth]{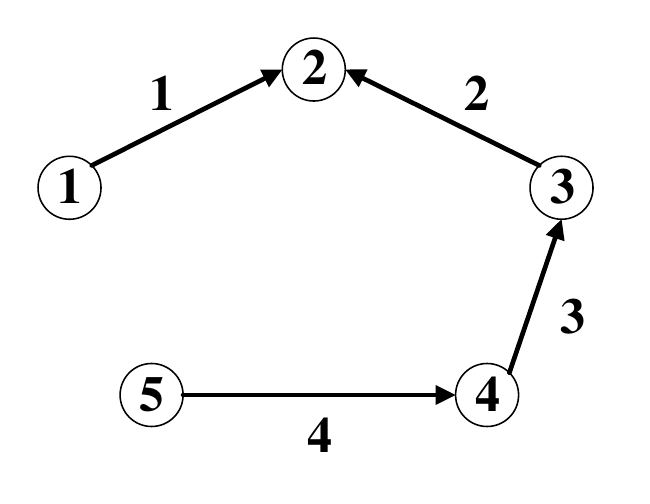}\label{graph}}
    \caption{The interaction graph $\mathcal{G}$: (a) without an orientation (b) with an orientation.}
    \label{simulation}
\end{figure} 

\indent To proceed with the first simulation, we assign an arbitrary orientation to the graph $\mathcal{G}$ as it is shown in Figure \ref{graph}. The attitude initial conditions, for both schemes, are chosen as: $\hat{R}_1(0)=\mathcal{R}_\alpha(-\frac{\pi}{2},v)$, $\hat{R}_2(0)=\mathcal{R}_\alpha(\frac{\pi}{2},v)$, $\hat{R}_3(0)=\mathcal{R}_\alpha(-\frac{\pi}{2},v)$, $\hat{R}_4(0)=\mathcal{R}_\alpha(\frac{\pi}{2},v)$ and $\hat{R}_5(0)=\mathcal{R}_\alpha(-\frac{\pi}{2},v)$, with $v=[0~0~1]^T$. In addition, for \textit{Hybrid observer}, we choose the following initial conditions for the auxiliary variables: $\xi_k(0)=0$ for every $k\in\{1, 2, 3, 4\}$. Note that, according to these initial conditions, one has $\bar R_k(0)=\mathcal{R}_\alpha(\pi,v)$ and $\xi_k(0)=0$, $\forall k=\{1, 2, 3, 4\}$, which implies that $x(0) \in \Upsilon \setminus \mathcal{A}$ and $x_h(0) \in \Upsilon_h \setminus \mathcal{A}_h$. Based on the parameters set $\mathcal{P}$, we select the following parameters: $\Xi=\{0.08\pi\}$, $A=\text{diag}([5, 8.57, 12])$, $\gamma=1.9251$, $\delta=0.0030$, $u=[0~0.6455~0.7638]^T$ and $\Delta^*=5$. For the gain parameters, we pick: $k_R= 1.1$ and $k_{\xi}=5$. To simulate the \textit{Hybrid observer}, we have used the HyEQ Toolbox \cite{Sanfelice_matlab}.\\
\indent Figure \ref{R_bar} and Figure \ref{xi} depict the time evolution of the relative attitude error norms $|\bar R_k(t)|_I$, for both schemes, and the auxiliary variables $\xi_k(t)$, $\forall k=\{1, 2, 3, 4\}$, respectively, associated with each edge. Notice that, at $t=0$, the variables $\xi_k(t)$, $\forall k=\{1, 2, 3, 4\}$, jump from $0$ to $0.08\pi$ and then converge to zero as $t\rightarrow \infty$. Also, the relative attitude error norms $|\bar R_k(t)|_I$, for both schemes, converge to zero as $t\rightarrow \infty$.

\begin{figure}[h]
    \centering
    \includegraphics[width=1\linewidth]{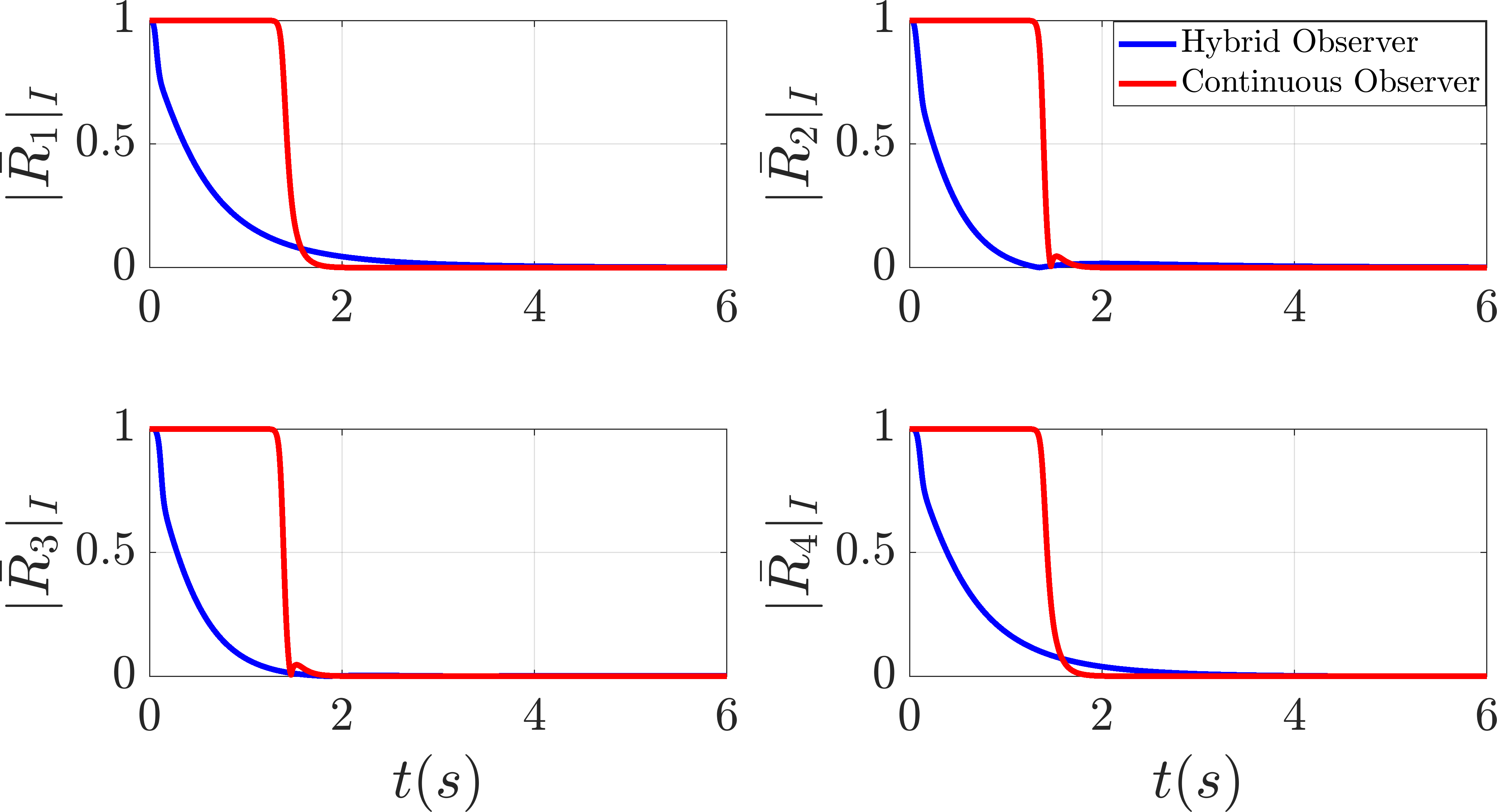}
    \caption{The time evolution of the relative attitude error norm, associated with each edge, for the  \textit{Continuous observer} and the \textit{Hybrid observer}.}
    \label{R_bar}
\end{figure}

\begin{figure}[h]
    \centering
    \includegraphics[width=0.8\linewidth]{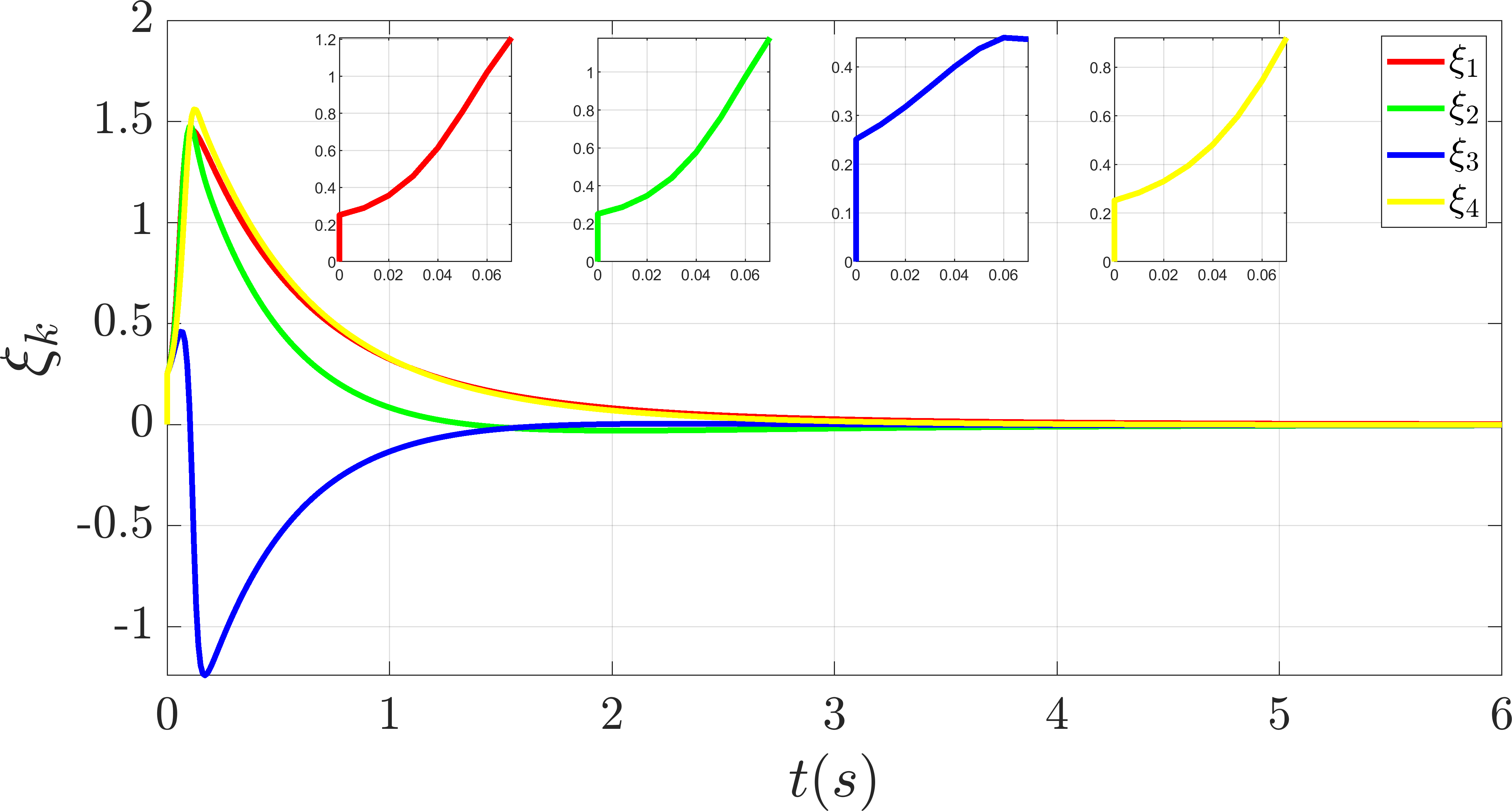}
    \caption{The time evolution of the hybrid variable $\xi_k$ associated with each edge.}
    \label{xi}
\end{figure}

\indent In our second simulation, we consider the proposed hybrid distributed pose estimation scheme \eqref{p_obs_f}-\eqref{p_obs_j}. We pick the same initial conditions as in the first simulation, along with the following initial conditions for the estimated positions: $\hat p_1(0)=[1~1~0]^T$, $\hat p_2(0)=[-1~2~1]^T$, $\hat p_3(0)=[-2~0~-1]^T$, $\hat p_4(0)=[-1~2~2]^T$ and $\hat p_5(0)=[-1~1~1]^T$. We pick $k_p=1$.\\
The time evolution of the position and the relative position estimation error norms are provided in Figure \ref{p_tilde} and Figure \ref{e}, respectively.

\begin{figure}[H]
    \centering
    \includegraphics[width=0.8\linewidth]{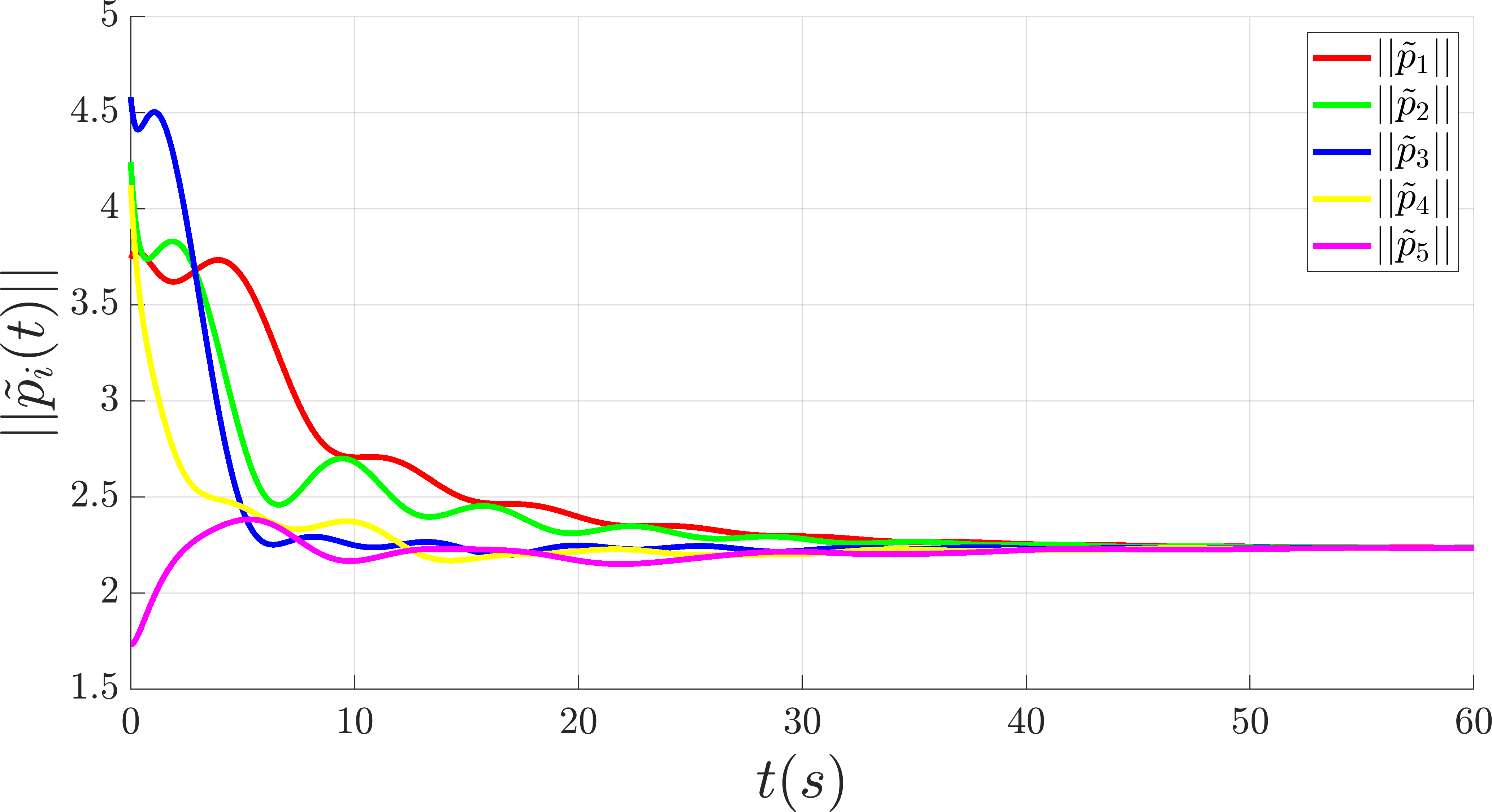}
    \caption{The time evolution of the position estimation error norm.}
    \label{p_tilde}
\end{figure}

\begin{figure}[H]
    \centering
    \includegraphics[width=0.8\linewidth]{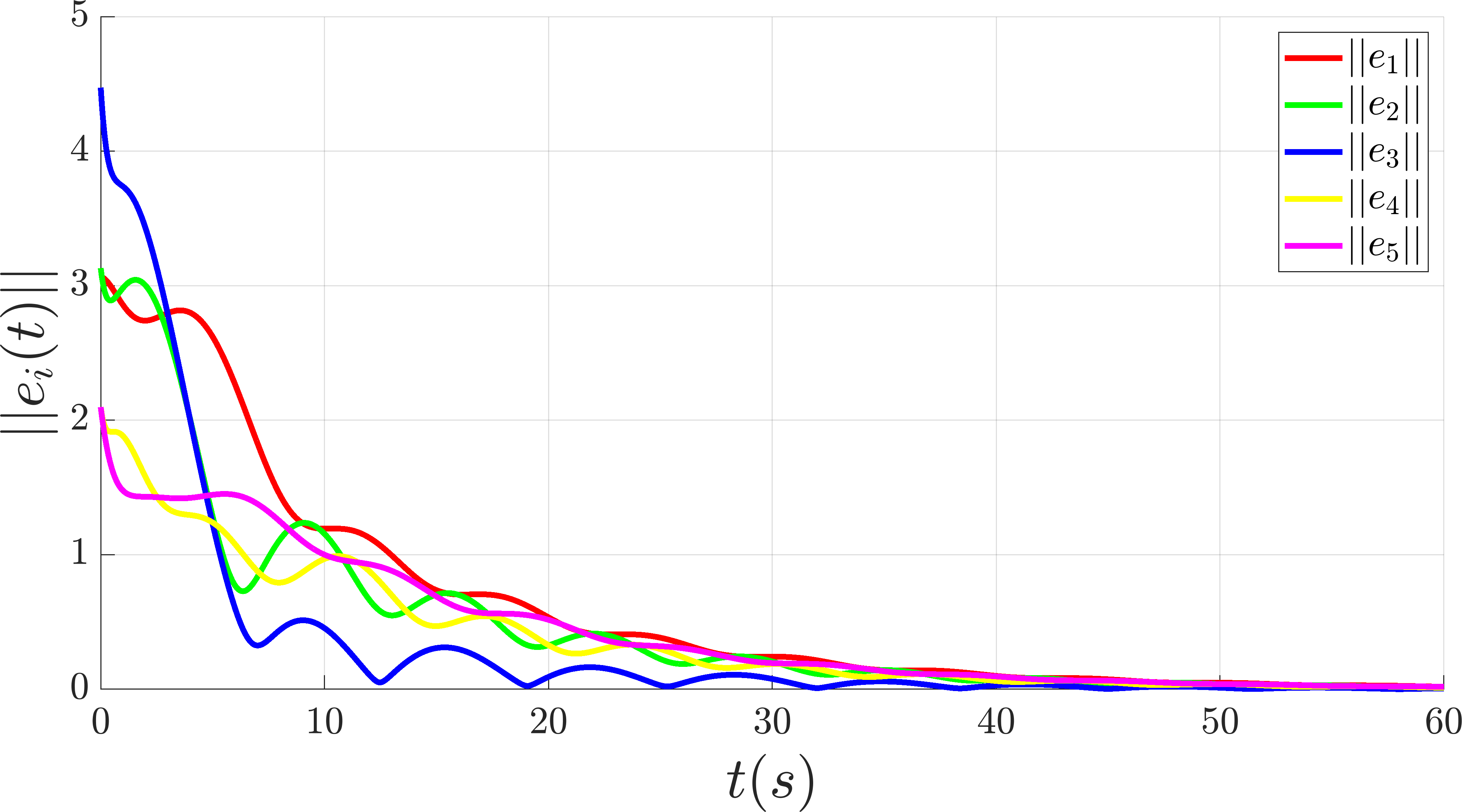}
    \caption{The time evolution of the relative position estimation error norm.}
    \label{e}
\end{figure}

\section{CONCLUSIONS}\label{s7}
Two nonlinear distributed attitude estimation schemes have been proposed for multi-agent systems evolving on $SO(3)$. The first continuous observer, endowed with almost global asymptotic stability, is used as a baseline for the derivation of a stronger hybrid version enjoying global asymptotic stability of the desired equilibrium set $\mathcal{A}_h$, which implies that the attitude of each agent can be estimated (globally) up to a common constant orientation which can be uniquely determined if at least one agent has access to its absolute attitude. This hybrid distributed attitude estimation scheme relies on auxiliary time-varying scalar variables associated to each edge $k$, namely $\xi_k$, which are governed by the hybrid dynamics \eqref{H_observer_dynamics_1}-\eqref{H_observer_dynamics_2}. These auxiliary variables are appropriately designed to keep the relative attitude errors away from the undesired equilibrium set $\Upsilon_h\setminus \mathcal{A}_h$ generated by smooth vector fields. Furthermore, the hybrid distributed attitude estimation scheme has been used with a hybrid distributed position estimation scheme to globally asymptotically estimate the pose of $N$ vehicles navigating in three-dimensional space, up to a common constant translation and orientation, relying on local relative time-varying bearings between the agents, as well as the individual linear velocity measurements.\\ 
%The first scheme is a continuous nonlinear distributed attitude estimation scheme designed directly on $SO(3)$. However, this scheme does not guarantee the global convergence to the desired equilibrium set $\mathcal{A}$ due to the topological obstruction of the rotational manifold $SO(3)$, discussed earlier, that prevents achieving global results with a smooth vector field. Therefore, in the second scheme, we employed a hybrid framework to design a hybrid nonlinear distributed attitude estimation scheme on $SO(3)\times \mathbb{R}$ that guarantees the global asymptotic stability of desired equilibrium set $\mathcal{A}_h$, which implies that the attitude of each agent will be estimated up to a common constant orientation which can be uniquely determined if at least one agent has access to its absolute attitude. This scheme relies on some auxiliary time-varying scalar variables, namely $\xi_k$, which are governed by hybrid dynamics \eqref{H_observer_dynamics_1}-\eqref{H_observer_dynamics_2}, where each variable is associated with one edge. These auxiliary variables are appropriately designed to keep the relative attitude errors away from the undesired equilibrium set $\Upsilon_h\setminus \mathcal{A}_h$ generated by smooth vector fields.\\ 
Note that our proposed estimation schemes rely on the assumption that the interaction graph topology is a tree (Assumption \ref{graph_ass}), which is practical in terms of communication and sensing costs. However, the main drawback of this graph topology is its vulnerability to failure since the failure of one agent will engender the disconnection of successive agents. Relaxing this assumption would be an interesting extension of this work. 
%In addition, extending our proposed hybrid scheme to address the attitude synchronization problem of multiple rigid bodies on $SO(3)$ is another future direction worthy of investigating.

%\addtolength{\textheight}{-10cm}   % This command serves to balance the column lengths
                                  % on the last page of the document manually. It shortens
                                  % the textheight of the last page by a suitable amount.
                                  % This command does not take effect until the next page
                                  % so it should come on the page before the last. Make
                                  % sure that you do not shorten the textheight too much.

%%%%%%%%%%%%%%%%%%%%%%%%%%%%%%%%%%%%%%%%%%%%%%%%%%%%%%%%%%%%%%%%%%%%%%%%%%%%%%%%

%%%%%%%%%%%%%%%%%%%%%%%%%%%%%%%%%%%%%%%%%%%%%%%%%%%%%%%%%%%%%%%%%%%%%%%%%%%%%%%%
\bibliographystyle{IEEEtran}
\bibliography{References}
\end{document}